\documentclass[11pt]{article}

\usepackage{geometry}
\usepackage{fullpage}
\usepackage{amsmath,amsfonts,amsthm,cite}
\usepackage{hyperref}
\usepackage{subcaption}
\usepackage{overpic}
\usepackage{xspace}
\usepackage{xargs}
\usepackage{graphicx}
\usepackage[labelfont=bf,font=small]{caption}
\usepackage[bottom]{footmisc}
\usepackage[section]{placeins}

\usepackage[colorinlistoftodos]{todonotes}

\newcommand{\bemph}[1]{\textbf{\textit{\boldmath #1}}}
\let\defn=\bemph

\newtheorem{theorem}{Theorem}[section]
\newtheorem{lemma}[theorem]{Lemma}
\newtheorem{corollary}[theorem]{Corollary}

\newcommand{\figref}[1]{Figure~\ref{#1}}
\newcommand{\thmref}[1]{Theorem~\ref{#1}}
\newcommand{\lemref}[1]{Lemma~\ref{#1}}

\newcommand{\PSPACE}{\textsc{PSPACE}\xspace}

{\makeatletter
 \gdef\xxxmark{%
   \expandafter\ifx\csname @mpargs\endcsname\relax % in minipage?
     \expandafter\ifx\csname @captype\endcsname\relax % in figure/caption?
       \marginpar{xxx}% not in a caption or minipage, can use marginpar
     \else
       xxx % notice trailing space
     \fi
   \else
     xxx % notice trailing space
   \fi}
 \gdef\xxx{\@ifnextchar[\xxx@lab\xxx@nolab}
 \long\gdef\xxx@lab[#1]#2{\textbf{[\xxxmark #2 ---{\sc #1}]}}
 \long\gdef\xxx@nolab#1{\textbf{[\xxxmark #1]}}
 \long\gdef\xxx@lab[#1]#2{}\long\gdef\xxx@nolab#1{}%
}

{\makeatletter \gdef\fps@figure{!htbp}}

\graphicspath{{figures/}}

\title{Walking through Doors is Hard, even without Staircases: \\
  Universality and PSPACE-hardness of Planar Door Gadgets%
  \thanks{A preliminary version of this paper appeared at FUN 2020~\cite{doors-fun2020}.}}
\author{%
  MIT Gadgets Group%
     \thanks{Artificial first author to highlight that the other authors (in
       alphabetical order) worked as an equal group. Please include all
       authors (including this one) in your bibliography, and refer to the
       authors as “MIT Gadgets Group” (without “et al.”).}
\and
  Jeffrey Bosboom%
    \thanks{MIT Computer Science and Artificial Intelligence Laboratory,
      32 Vassar St., Cambridge, MA 02139, USA.
      \protect\url{{bosboom,edemaine,diomidova,della,hayastl,jaysonl}@mit.edu}}
\and
  Erik D. Demaine\footnotemark[3]
\and
  Jenny Diomidova\footnotemark[3]
\and
  Della Hendrickson\footnotemark[3]
\and
  Hayashi Layers\footnotemark[3]
\and
  Jayson Lynch\footnotemark[3]
}

\date{}

\begin{document}

\maketitle

\begin{abstract}
  An open--close door gadget has two states and three tunnels that can be traversed by an
  agent (player, robot, etc.):
  the ``opening'' and ``closing'' tunnels set the gadget's state to open and
  closed, respectively, while the ``traverse'' tunnel can be traversed if and
  only if the door is in the open state.
  We prove that it is \PSPACE-complete to decide whether an agent can move from
  one location to another through a \emph{planar} system of any such door
  gadget, removing the traditional need for crossover gadgets and thereby
  simplifying past \PSPACE-hardness proofs of Lemmings and Nintendo games
  Super Mario Bros., Legend of Zelda, and Donkey Kong Country.
  Even stronger, we show that \emph{any gadget} in the
  motion-planning-through-gadgets framework can be simulated by a planar
  system of door gadgets: the open--close door gadget is a universal gadget.

  We prove that these results hold for a variety of door gadgets.
  In particular, the opening, closing, and traverse tunnel locations can have
  an arbitrary cyclic order around the door; each tunnel can be directed
  or undirected; and the opening tunnel can instead be an optional button
  (with identical entrance and exit locations).
  Furthermore, we show the same hardness and universality results for
  two simpler types of door gadgets.
  A \emph{self-closing} door gadget has two states and only two tunnels:
  ``opening'' which behaves like a door, and
  ``self-closing'' which can be traversed
  only when the door is open and whose traversal changes the state to closed
  (like a traverse tunnel followed by a closing tunnel).
  A \emph{symmetric} self-closing door also has two states and two tunnels:
  ``self-closing'' as above, and ``self-opening'' which can be traversed
  only when the door is closed and whose traversal changes the state to open;
  thus the states and tunnels are symmetric under an ``open''/``close'' swap.
  Again we show that any self-closing door gadget planarly simulates any
  gadget, and thus the reachability motion planning problem is
  \PSPACE-complete.
  Then we apply this framework to prove new \PSPACE-hardness results for
  eight different 3D Mario video games and Sokobond.
  %, and Katamari Damacy.
\end{abstract}

\section{Introduction}

%% OUTLINE
% Nintendoor paper
% Gadget framework
%    Only reversible gadgets were proven \PSPACE-hard (and maybe I/O gadgets; which paper will get finished first?)

Puzzle video games are rife with doors that block the player's passage
when closed/locked.  To open such a door, the player often needs to collect
the right key or keycard, or to press the right combination of buttons or
pressure plates, or to solve some other puzzle.  Many of these game features
in sufficient generality imply that the video game is NP-hard or \PSPACE-hard,
according to a series of ``metatheorems'' \cite{platform,gaming,Bloxorz}.
%(including papers at FUN 2010 and 2012).

An intriguing twist is to use doors as a framework for proving hardness of
video games that do not ``naturally'' have doors, but have some mechanics that
suffice to simulate doors via a small construction or ``gadget''.
Our research was initially motivated by
a local ``door gadget'' introduced in Viglietta's FUN 2014 paper
\cite{lemmings} proving Lemmings \PSPACE-complete.
This door gadget is a portion of a level design containing three directed paths
that the player can traverse: a ``traverse'' path that can be traversed if and
only if the door is open, a ``closing'' path that forces the door to close,
and an ``opening'' path that allows the player to open the door if desired.
%We call such a gadget an \defn{open-buttoned door}.
Viglietta \cite[Metatheorem~3]{lemmings} proved that such a door gadget,
together with the ability to wire together door entrance/exit locations according
to an arbitrary graph (including crossovers for a 2D game like Lemmings),
where the player has the choice of how to traverse the graph,
suffice to prove \PSPACE-hardness of deciding whether an agent can move from
one location to another.
At the same FUN (2014), Aloupis et al.~\cite{nintendoor}
used this door framework to prove Legend of Zelda: Link to the Past
and Donkey Kong Country 1, 2, and 3 \PSPACE-complete.
At the next FUN (2016), Demaine et al.~\cite{demaine2016super}
used this door framework to prove Super Mario Bros.\ \PSPACE-complete.
All of these reductions also build a crossover gadget for wiring paths between
door gadgets.

In this paper, we show that such crossover gadgets are unnecessary:
just building the door gadget, together with the ability to connect together
the entrance/exit locations in a planar way,
suffices to prove \PSPACE-hardness.
%We show that the same result holds for a variety of different door gadgets.
%In particular, this simplifies the \PSPACE-hardness proofs of Lemmings \cite{lemmings}, Legend of Zelda: Link to the Past \cite{nintendoor}, Donkey Kong Country 1, 2, and 3 \cite{nintendoor}, and Super Mario Bros.~\cite{demaine2016super}, because we can now omit the crossover gadgets from those reductions.
Furthermore, we show that door gadgets have a much stronger \emph{universality}
property, in a general framework which we now introduce.

\subsection{Motion Planning through Gadgets}

The \defn{motion-planning-through-gadgets framework}
was initiated at FUN 2018 \cite{gadgets}
and developed in a series of papers
\cite{gadgets2,pullingBlocks,push1f,IOgadgets,GadgetsVictory,demaine2025pspace,lynch-thesis,della-masters,ani-masters}.
It formalizes the idea of one agent (or more)
moving through a graph of local gadgets, where each gadget has local state and
possible traversals (depending on the state),
and whose traversal affects that gadget's state (only).
\figref{fig:state-diagram} shows an example of a gadget
in this framework.
In general, a \defn{gadget} consists of three finite sets ---
\defn{states}~$Q$, \defn{locations}~$L$, and \defn{transitions}~$T$ ---
where each transition $\in T$ is of the form $(q,a) \to (r,b)$
where $q,r \in Q$ and $a,b \in L$.
The meaning of $(q,a) \to (r,b) \in T$ is that, when the gadget is in state~$q$,
the agent can enter the gadget at location~$a$,
change the gadget's state to~$r$,
and then exit the gadget at location~$b$;
we call this agent action a \defn{traversal} $a \to b$.
A \defn{traversal sequence} $[a_1\to b_1,\dots,a_k\to b_k]$
on the locations $L$ is \defn{legal} from state $q_0$
if there is a corresponding sequence of transitions
$[(q_0,a_1) \to (q_1,b_1), \dots, (q_{k-1},a_k) \to (q_k,b_k)]$,
where each start state of each transition matches the end state
of the previous transition ($q_0$ for the first transition).

\begin{figure}
  \centering
  \includegraphics{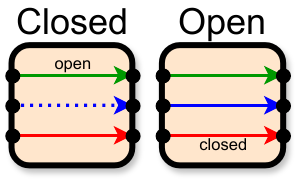}
  \caption{State diagram of a 2-state 6-location 5-transition gadget
    (the directed open-tunneled open--close door
    from \figref{fig:different-doors-directed}).
    The left and right copies represent the two states
    (``closed'' and ``open'').
    In each state, we draw the allowed transitions as arrows,
    labeled with the resulting state change, if any.
    Dotted transitions are forbidden and are just drawn for intuition
    (and to clarify locations).
  }
  \label{fig:state-diagram}
\end{figure}

A \defn{system} of gadgets consists of a finite set of gadgets
(which are often all instances of a single gadget type~$G$),
an initial state for each gadget, and a \defn{connection graph} whose vertices
are the locations of all the gadgets.
The agent can freely traverse edges of the connection graph
without any state change, so each connected component
of the connection graph acts like a single location of the system.
The \defn{reachability problem} asks, given a system of gadgets
and two locations $s,t \in L$, whether the agent can move from $s$ to $t$
via a sequence of connection-graph edges and legal gadget traversals.
%where the subsequence of traversals
%through any particular gadget is legal as defined above.
More precisely, a \defn{system traversal} $a_1 \to^* b_k$
is a sequence of traversals $[a_1\to b_1,\dots,a_k\to b_k]$,
each on a potentially different gadget in the system,
where the connection graph has a path from $b_i$ to $a_{i+1}$ for each~$i$;
and it is \defn{legal} if the subsequence of all traversals through any
particular gadget $G$ is legal as defined above,
from the initial state of $G$ assigned by the system,
for every $G$ in the system.

In \defn{planar motion planning},
each gadget additionally has a specified cyclic ordering of its vertices,
and the system of gadgets is embedded in the plane without intersections.
More precisely, a system of gadgets is \defn{planar} if the following
construction produces a planar graph:
(1)~replace each gadget with a wheel graph,
which has a cycle of vertices corresponding to the locations on the gadget
in the appropriate order, and a central vertex connected to each location;
and (2)~connect locations on these wheels with edges
according to the connection graph.
In \defn{planar reachability},
we restrict to planar systems of gadgets.
Note that this definition allows rotations and reflections of gadgets,
but no other permutation of their locations.

%In the 1-player unbounded setting considered here,
%the original work \cite{gadgets,gadgets2} analyzed gadgets that are
%\begin{enumerate}
%\item \defn{deterministic}, meaning that when an agent enters a gadget at any
%  location, it has a unique exit location and causes a unique state change;
%\item \defn{reversible}, meaning that every such traversal can be
%  immediately undone, both in terms of agent location and gadget state change;
%  and
%\item \defn{$k$-tunnel}, meaning that the $2k$ entrance/exit locations can be
%  paired up such that, in any state, traversal paths only connected paired
%  locations (in some direction).
%\end{enumerate}
%Restricted to deterministic reversible $k$-tunnel gadgets,
%Demaine et al.~\cite{gadgets2} characterized which gadget sets make
%reachability \PSPACE-complete:
%whenever the gadget set contains a gadget with \defn{interacting
%tunnels}, meaning that traversing some traversal path changes (adds or removes)
%the traversability of some other traversal path (in some direction).
%Furthermore, they proved the same characterization for planar reachability,
%obviating the need for a crossover gadget.

As shown in \figref{fig:state-diagram},
door gadgets naturally fit into this motion-planning-through-gadgets framework.
(Indeed, doors were one of the inspirations for the framework.)
There are also several \PSPACE-completeness results for reachability and/or
planar reachability with various families of gadgets
\cite{gadgets,gadgets2,pullingBlocks,push1f,IOgadgets,GadgetsVictory,demaine2025pspace,lynch-thesis,della-masters,ani-masters},
but none of them apply to door gadgets.
For example, the original two papers \cite{gadgets,gadgets2}
characterized the complexity of [planar] reachability for gadgets that are
(1)~\defn{deterministic}, meaning that when an agent enters a gadget at some
location, it has a unique exit location and causes a unique state change;
(2)~\defn{reversible}, meaning that every transition can be immediately undone;
and
(3)~\defn{$k$-tunnel}, meaning that the $2k$ locations can be
paired up such that all traversals move the agent between paired locations
(in some direction).
Some door gadgets, such as the one in \figref{fig:state-diagram},
are both deterministic and $k$-tunnel, but none of them are reversible:
a door cannot be reopened by traversing the closing tunnel backwards.
%Notably, however, the door gadget used in
%\cite{lemmings,nintendoor,demaine2016super}
%is neither deterministic (the open path can open the door or not, according
%to the player's choice) nor reversible (the tunnels are all directed in fixed
%directions), so the existing characterization and planarity result
%do not apply.
Nonetheless, we will show \PSPACE-completeness of planar reachability
with doors, along with a stronger property called ``universality'',
defined in terms of simulations.
Here we follow the definitions of \cite{push1f,IOgadgets,ani-masters};
see also \cite{della-masters}.

A set $\mathcal G$ of gadgets is [planarly] \defn{universal}
for a class $\mathcal C$ of gadgets
if $\mathcal G$ can [planarly] simulate every gadget in~$\mathcal C$,
as defined next.
(Typically, we also require that the simulating gadgets $\mathcal G$
are in the class, i.e., $\mathcal G \subseteq \mathcal C$.)

A set $\mathcal G$ of gadgets can \defn{simulate} another gadget $H$
if it is possible to build a system of gadgets in $\mathcal G$ that
``acts like'' $H$ when we restrict attention to a subset of locations.
More precisely, a \defn{simulation} of a gadget $H$ in state $q$
by gadgets in $\mathcal G$
consists of a system of gadgets in $\mathcal G$, and an injective mapping $m$
from every location of $H$ to a distinct location in the system, such that
a traversal sequence $[a_1 \to b_1, \dots, a_k \to b_k]$
on the locations in $H$ is legal from state $q$
if and only if there exists a sequence of system traversals
$m(a_1) \to^* m(b_1), \dots, m(a_k) \to^* m(b_k)$
that is legal in the sense above:
the subsequence of all traversals through any
particular gadget $G$ is legal,
from the initial state of $G$ assigned by the system,
for every $G$ in the system.
A \defn{planar simulation} of a gadget $H$
is a planar system and mapping $m$ with the same properties,
and where the cyclic order of locations of $H$ maps via $m$ to
locations in cyclic order around the outside face of the system.
A [planar] simulation of an entire gadget $H$ consists of
a [planar] simulation of $H$ in state $q$, for all states $q$ of~$H$,
that differ only in their assignments of initial states.

A [planar] simulation of a gadget $H$ using gadgets in $\mathcal G$
in particular provides a reduction from [planar] reachability
with $\mathcal G \cup \{H\}$ to [planar] reachability with $\mathcal G$.
Thus, if [planar] reachability is \PSPACE-hard with
$\mathcal G \cup \{H\}$, then it is also \PSPACE-hard with $\mathcal G$.

%A finite set $\mathcal G$ of gadgets \defn{[planarly] simulates}
%a gadget $G$ if there is a [planar] simulation of $G$
%using only gadgets in $\mathcal G$.

\subsection{Door Gadgets}

In this paper, we develop a motion-planning-through-\emph{doors} framework,
completing another subspace of the motion-planning-through-gadgets framework.
Our framework applies to a variety of different door gadgets, including the
door gadget of \cite{lemmings,nintendoor,demaine2016super}
and other examples in \figref{fig:different-doors}.
Typically, an \defn{open--close door gadget}%
\footnote{The conference version of this paper \cite{doors-fun2020}
  referred to open--close doors as simply ``doors'', but this made it difficult
  to distinguish from other types of (self-closing) doors we also introduce.
  The more descriptive name ``open-close doors'' is also used in
  \cite{demaine2016super,janggi}
  (but \cite{lemmings,nintendoor} use ``doors'').}
has two states (``open'' and ``closed'')
and three tunnels: ``traverse'', ``closing'', and ``opening''.
Each tunnel may be individually \defn{directed} (traversable in one direction) or
\defn{undirected} (traversable in both directions).
In addition, the opening traversal may have identical entrance and exit
locations, turning it into an opening \defn{button} instead of a tunnel.
In this case, the traversal changes the door's state but does not
move the agent (breaking the $k$-tunnel assumption).

\begin{figure}
  \centering
  \subcaptionbox{\label{fig:different-doors-directed} A directed open-tunneled open/\allowbreak close door.}{\includegraphics[scale=1]{different-doors-directed}}
  \hfill
  \subcaptionbox{\label{fig:different-doors-undirected} An undirected open-tunneled open/\allowbreak close door.}{\includegraphics[scale=1]{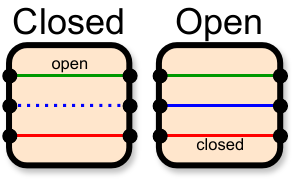}}
  \hfill
  \subcaptionbox{\label{fig:different-doors-mixed} A mixed open-buttoned open/\allowbreak close door.}{\includegraphics[scale=1]{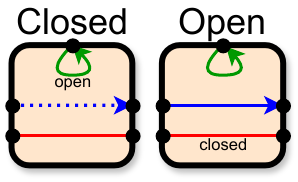}}
  \caption{State diagrams for a few different open--close doors.
    The opening button/tunnel is green,
    the traverse tunnel is blue, and the closing tunnel is red.
  }
  \label{fig:different-doors}
\end{figure}

It does not make sense to talk about closing or traverse buttons, because the agent will never have an incentive to use them.
Similarly, we can assume without loss of generality that traversing the opening tunnel/button deterministically opens the door, because it is never beneficial to leave the door closed when you have the option to open it.

%In Section~\ref{sec:scd},
We introduce two more families of door gadgets,
illustrated in \figref{fig:self-closing-doors}.
A \defn{self-closing door} has two states
but only two tunnels: ``opening'' and ``self-closing''.
The self-closing traversal is possible only in the open state,
and it forcibly changes the state to closed.
As before, each tunnel can be either directed or undirected;
and the opening traversal forces the state to open,
but we allow the opening traversal to have identical start and end
locations, which effectively allows optional opening.
A \defn{symmetric self-closing door} has two states and two tunnels:
``self-opening'' and ``self-closing''.
The self-opening/self-closing traversal is possible only in the closed/open state,
respectively, and it forcibly changes the state to open/closed, respectively.
(This definition is fully symmetric between ``open'' and ``close''.)

\begin{figure}
  \centering
  \subcaptionbox{\label{fig:self-closing-doors-undirected} An undirected open-tunneled self-closing door.}{\includegraphics[scale=1]{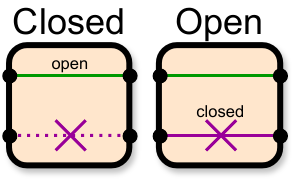}}
  \hfill
  \subcaptionbox{\label{fig:self-closing-doors-directed-button} A directed open-buttoned self-closing door.}{\includegraphics[scale=1]{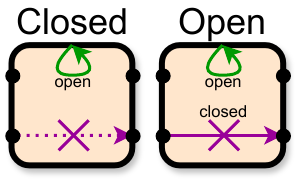}}
  \hfill
  \subcaptionbox{\label{fig:self-closing-doors-mixed} A mixed self-closing door.}{\includegraphics[scale=1]{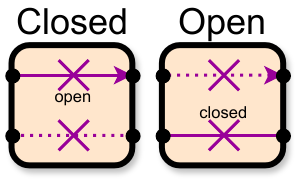}}
  \caption{State diagrams for a few different self-closing doors.
    The opening button/tunnel is green,
    and self-opening/self-closing tunnel(s) are purple
    with an X (to indicate that they close themselves when traversed).
    Dotted traversals are forbidden.
  }
  \label{fig:self-closing-doors}
\end{figure}

To illustrate how these various door gadgets work,
and how they could be built in an actual puzzle game,
\figref{fig:sokobond-doors} shows how to implement one of each in
a pushing-block game that supports pushing a polyomino piece.
More precisely, a $1 \times 1$ agent can traverse the unit-square grid,
polyominoes (indicated by a string of circled Hs and Os) can be pushed
by the agent in the agent's direction of motion, and
square blocks are fixed (cannot be pushed).

\begin{figure}
  \centering
  \footnotesize

  \subcaptionbox{\label{fig:sokobond-door} Undirected open-buttoned open--close door, open state}{%
    \hspace*{4.5em}%
    \begin{overpic}[scale=0.6]{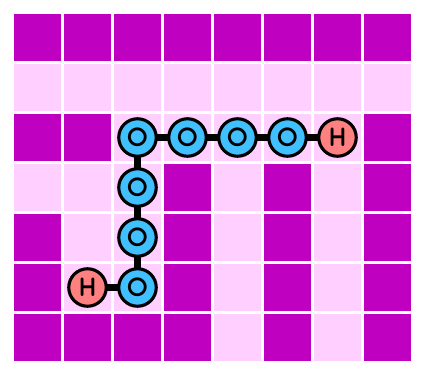}
      \put(12,67){\makebox(0,0)[r]{\strut traverse in}}
      \put(89,67){\makebox(0,0)[l]{\strut traverse out}}
      \put(12,44){\makebox(0,0)[r]{\strut opening button}}
      \put(56,25){\makebox(0,0){\rotatebox{90}{\strut closing in}}}
      \put(80,25){\makebox(0,0){\rotatebox{90}{\strut closing out}}}
    \end{overpic}%
    \hspace*{4.5em}%
  }
  \hfil
  \subcaptionbox{Undirected open-buttoned open--close door, closed state}{%
    \hspace*{4.5em}%
    \begin{overpic}[scale=0.6]{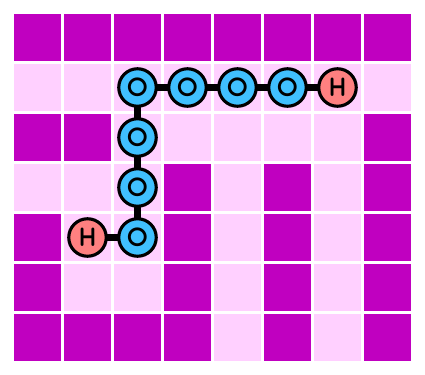}
      \put(12,67){\makebox(0,0)[r]{\strut traverse in}}
      \put(89,67){\makebox(0,0)[l]{\strut traverse out}}
      \put(12,44){\makebox(0,0)[r]{\strut opening button}}
      \put(56,25){\makebox(0,0){\rotatebox{90}{\strut closing in}}}
      \put(80,25){\makebox(0,0){\rotatebox{90}{\strut closing out}}}
    \end{overpic}%
    \hspace*{4.5em}%
  }

  \medskip

  \subcaptionbox{Directed open-buttoned self-closing door, open state}{\vbox{%
    \vspace*{3.0em}%
    \hbox{%
      \hspace*{4.8em}%
      \begin{overpic}[scale=0.6]{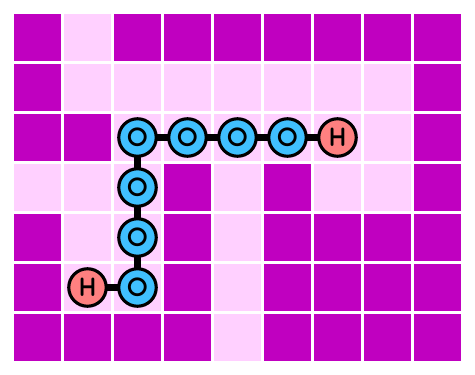}
        \put(19,80){\makebox(0,0){\rotatebox{90}{\strut self-closing in}}}
        \put(11,39){\makebox(0,0)[r]{\strut opening button}}
        \put(50.5,18){\makebox(0,0){\rotatebox{90}{\strut self-closing out}}}
      \end{overpic}%
      \hspace*{4.8em}%
    }%
    \vspace*{0.4em}%
  }}
  \hfil
  \subcaptionbox{Directed open-buttoned self-closing door, closed state}{\vbox{%
    \vspace*{3.0em}%
    \hbox{%
      \hspace*{4.8em}%
      \begin{overpic}[scale=0.6]{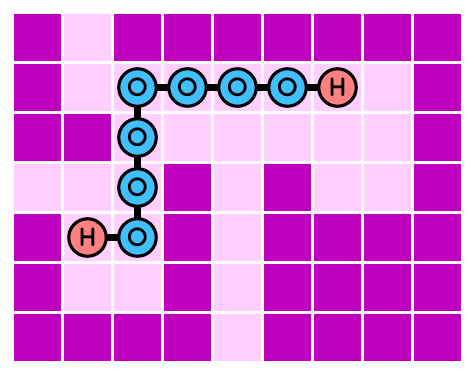}
        \put(19,80){\makebox(0,0){\rotatebox{90}{\strut self-closing in}}}
        \put(11,39){\makebox(0,0)[r]{\strut opening button}}
        \put(50.5,18){\makebox(0,0){\rotatebox{90}{\strut self-closing out}}}
      \end{overpic}%
      \hspace*{4.8em}%
    }%
    \vspace*{0.4em}%
  }}

  \medskip

  \subcaptionbox{Undirected symmetric self-closing door, open state}{\vbox{%
    \vspace*{2.8em}%
    \hbox{%
      \hspace*{1.2em}%
      \begin{overpic}[scale=0.6]{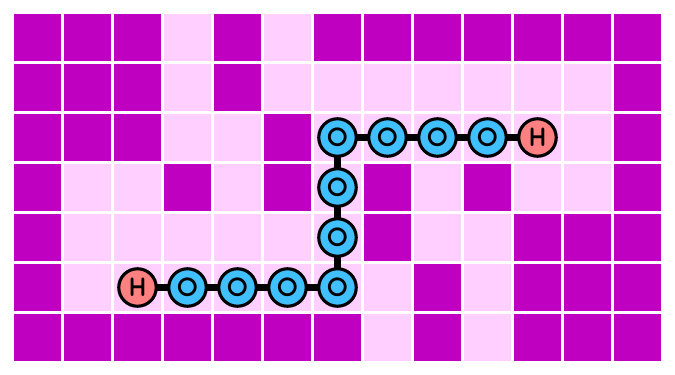}
        \put(28,53){\makebox(0,0){\rotatebox{90}{\strut self-opening out}}}
        \put(43,56){\makebox(0,0){\rotatebox{90}{\strut self-closing in}}}
        \put(58,-1.5){\makebox(0,0){\rotatebox{90}{\strut self-opening in}}}
        \put(73,-1){\makebox(0,0){\rotatebox{90}{\strut self-closing out}}}
      \end{overpic}%
      \hspace*{1.2em}%
    }%
    \vspace*{3.2em}%
  }}
  \hfil
  \subcaptionbox{Undirected symmetric self-closing door, closed state}{\vbox{%
    \vspace*{2.8em}%
    \hbox{%
      \hspace*{1.2em}%
      \begin{overpic}[scale=0.6]{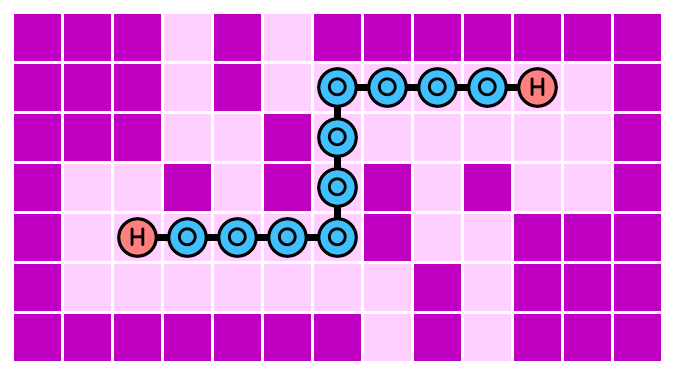}
        \put(28,53){\makebox(0,0){\rotatebox{90}{\strut self-opening out}}}
        \put(43,56){\makebox(0,0){\rotatebox{90}{\strut self-closing in}}}
        \put(58,-1.5){\makebox(0,0){\rotatebox{90}{\strut self-opening in}}}
        \put(73,-1){\makebox(0,0){\rotatebox{90}{\strut self-closing out}}}
      \end{overpic}%
      \hspace*{1.2em}%
    }%
    \vspace*{3.2em}%
  }}

  \caption{Implementation of an open--close door (a--b), self-closing door (c--d), and
    symmetric self-closing door (e--f) in any pushing-block game
    that allows pushing a polyomino piece
    (drawn here as a molecule, in the style of Sokobond \cite{Sokobond}).
    The symmetric self-closing door gadget is itself
    $180^\circ$ rotationally symmetric.}
  \label{fig:sokobond-doors}
\end{figure}

\subsection{Our Results}

Our main result is that every door gadget mentioned above
(open--close, self-closing, and symmetric self-closing,
of every variation) is planarly universal
for the class of \emph{all gadgets}.
These are the first examples of fully universal gadgets:
door gadgets are in this sense among the most powerful gadgets.
%Furthermore, the simulations can be made planar:
%every door gadget is planarly universal.
Section~\ref{sec:scd} proves this result for self-closing and
symmetric self-closing doors, and then Section~\ref{sec:doors}
extends the result to open--close doors,
in many cases by simulating self-closing doors.
Because we already know that planar reachability is \PSPACE-hard
for \emph{some} gadgets \cite{gadgets2},
we conclude that planar reachability is \PSPACE-hard for \emph{any} door gadget.

%In Section~\ref{sec:scd-universality}, we prove that every self-closing door and symmetric self-closing door is universal, meaning that any one of them can simulate \emph{all} possible motion-planning gadgets (and, therefore, motion-planning with any of these gadgets is \PSPACE-complete).
%This result provides the first examples of fully universal gadgets.
%In Section~\ref{sec:scd-planar-universality}, we show that these results also hold in the planar setting.
%
%In Section~\ref{sec:doors}, we formally define door gadgets, and prove that all of them are planarly universal (and, therefore, planar motion-planning with any door is \PSPACE-complete).

Our result should in particular make it easier to prove 2D games \PSPACE-hard:
it suffices to construct any single open--close door, self-closing door, or
symmetric self-closing door gadget,
and to show how to connect the door entrances/exits together in a planar graph.
For example, the crossover gadgets previously constructed for
Lemmings \cite[Figure~2(e)]{lemmings},
Legend of Zelda: Link to the Past and Donkey Kong Country 1, 2, and 3
\cite[Figures~28 and~20]{nintendoor},
and Super Mario Bros.~\cite[Figure~5]{demaine2016super}
are no longer necessary for those \PSPACE-hardness proofs:
they can now be omitted.
(See Section~\ref{sec:simplifications} for details.)
%
% Because of their reduced conceptual complexity --- only two traversal paths,
% which behave essentially identically for symmetric self-closing doors ---
% we have found it even easier to prove games \PSPACE-hard by building
% self-closing door gadgets.

In Section~\ref{sec:applications}, we apply this approach by proving new
\PSPACE-hardness results
for one 2D game, Sokobond, and eight different 3D Mario games:
Super Mario 64, Super Mario 64 DS, Super Mario Sunshine,
Super Mario Galaxy, Super Mario Galaxy 2, Super Mario 3D Land/World,
Super Mario Odyssey, and Captain Toad:\ Treasure Tracker
(and the associated levels in Super Mario 3D World).
These reductions consist of just one gadget ---
any one of the doors from \figref{fig:sokobond-doors} for Sokobond,
and a symmetric self-closing door for the 3D games ---
along with easy methods for connecting these gadgets.
For the 3D games, the main benefit is the simplicity of the symmetric
self-closing door: crossovers are generally easy in the 3D games, though
it remains convenient that we do not need to explicitly build them.

\subsection{Related Work}

Recently, we realized that Dor and Zwick \cite[Section~2]{sliding-door}
introduced one version of the door
(specifically, a directed open-buttoned open--close door)
back in 1999, under the name ``sliding door''.%
\footnote{We also can't help but notice the first author's appropriate name:
  Dorit Dor.}
They built this gadget in a generalized form of the pushing-block game Sokoban,
called SOKOBAN$^+$,
which we might today call PushPull?-2FS with $1 \times 2$ blocks
(see \cite{pullingBlocks} for relevant terminology).
They also built one other gadget we call a ``diode''
(see Section~\ref{sec:diode}),
and used these two gadgets to build a direct simulation of Turing machines,
proving \PSPACE-hardness of the game.
Along the way, they showed that two of their door gadgets and two
diode gadgets can build a crossover gadget.

Interpreted in the motion-planning-through-gadgets framework,
their construction shows \PSPACE-hardness of planar reachability with
the diode and one planar layout of the directed open-buttoned open--close door.
Specifically, the planar layout is Case 10: OTcCt of
Section~\ref{sec:door-2-10-12}, which is one of our easiest open--close door cases
because it can simulate a self-closing door by connecting the traverse
output location to the closing input location.
In fact, self-closing doors can in turn easily simulate a diode
(as shown in Lemma~\ref{lem:sim-diode}),
so their construction does not need a diode.
(This fact was not useful to Dor and Zwick \cite{sliding-door}
because they used their diode to build their door.)
%: just connect the self-closing exit to the open button
Thus, we obtain another proof of \PSPACE-hardness of planar reachability with
just the Case-10 open--close door.
%the directed open-buttoned self-closing door.
In fact, their simulation of an arbitrary Turing machine is relatively easy
to adapt into a simulation of an arbitrary gadget,
giving another proof of universality for this case.

\section{Self-Closing Doors}
\label{sec:scd}
In this section, we define two types of gadgets --- \bemph{self-closing doors} and \bemph{symmetric self-closing doors} --- and show that each of these gadgets is planarly universal and \PSPACE-complete. Unlike open--close doors, which have $5$ or $6$ locations, these gadgets only have $3$ or $4$, making analysis much simpler.

\subsection{Types of Self-Closing Doors}

A \bemph{self-closing door} is a $2$-state gadget that has a tunnel that closes itself when traversed
(the \emph{self-closing} tunnel), and a tunnel or button
that reopens said tunnel (the \emph{opening} tunnel/button).
A self-closing door is \bemph{open-tunneled} if it has an opening tunnel, and \bemph{open-buttoned} if it has an opening button.
A \bemph{symmetric self-closing door} is a $2$-state gadget that has two tunnels, exactly one of which is open at a time, and traversing either tunnel closes that tunnel and reopens the other tunnel.
\figref{fig:self-closing-doors} gives the state diagram for some self-closing doors.

Each of the (one or two) tunnels in a [symmetric] self-closing door can be either directed or undirected.
We call a [symmetric] self-closing door \bemph{directed} if all its tunnels are directed, \bemph{undirected} if all its tunnels are undirected, and \bemph{mixed} if it has both directed and undirected tunnels (in which case it cannot be an open-buttoned self-closing door).

Throughout this paper, the opening button/tunnel will be colored green, and a self-closing tunnel will be colored purple with an `X' over it.
A dotted line indicates a closed tunnel and a solid line indicates an open tunnel.
A tunnel without an arrowhead is undirected, while a tunnel with a solid arrowhead is directed (in the direction of the arrowhead).
Sometimes a construction works both for directed and undirected tunnels, in which case we use hollow arrowheads.
A tunnel with a hollow arrowhead may be undirected or directed in that direction.
A tunnel with hollow arrowheads on both ends may be undirected or directed in either direction.

First we show that these many types of [symmetric] self-closing doors all planarly simulate one ``simplest'' self-closing door, so for the purposes of universality and hardness, we can focus on analyzing just it.
%proving it planarly universal and PSPACE-complete.

\begin{lemma}\label{lem:canonical-scd}
  Any self-closing door or symmetric self-closing door planarly simulates the directed open-buttoned self-closing door.
\end{lemma}

\begin{proof}
  First, we can simulate some open-buttoned self-closing door (directed or undirected) as shown in \figref{fig:scd-to-open-buttoned}. If it is directed, we are done. If it is undirected, we can simulate the directed open-buttoned self-closing door as shown in \figref{fig:scd-to-directed}.
\end{proof}

\begin{figure}
  \centering
  \includegraphics[scale=.8]{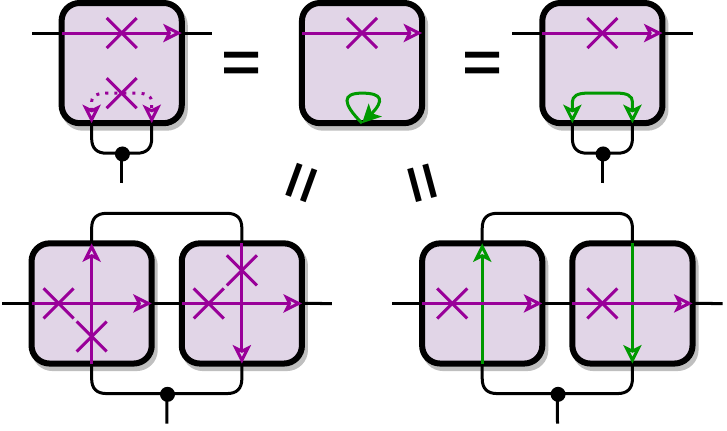}
  \caption{Any self-closing door or symmetric self-closing door planarly simulates an \emph{open-buttoned} self-closing door.}
  \label{fig:scd-to-open-buttoned}
\end{figure}

\begin{figure}
  \centering
  \includegraphics[scale=.8]{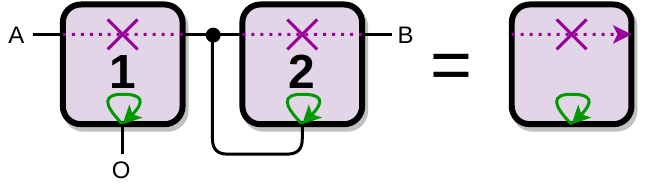}
  \caption{The \emph{undirected} open-buttoned self-closing door planarly simulates the \emph{directed} open-buttoned self-closing door.}
  \label{fig:scd-to-directed}
\end{figure}

From now on, we will refer to the directed open-buttoned self-closing door as simply ``the self-closing door''. Any result of the form ``the self-closing door simulates $G$'' can be extended to arbitrary self-closing and symmetric self-closing doors using \lemref{lem:canonical-scd}.

\subsection{Diode}
\label{sec:diode}

Next we show how to build a \defn{diode},
which is a 1-state gadget with a single directed tunnel.

\begin{lemma}\label{lem:sim-diode}
  The self-closing door planarly simulates a diode.
\end{lemma}

\begin{proof}
  See \figref{fig:scd-diode}.
\end{proof}

\begin{figure}
  \centering
  \begin{subfigure}{0.45\textwidth}
    \centering
    \includegraphics[scale=.8]{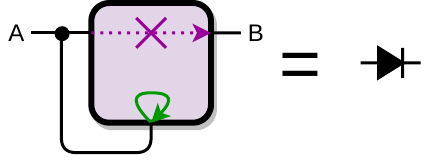}
    \caption{Simulating a diode.}
    \label{fig:scd-diode}
  \end{subfigure}
  \begin{subfigure}{0.45\textwidth}
    \centering
    \includegraphics[scale=.8]{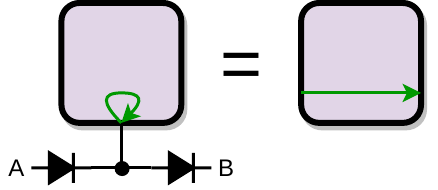}
    \caption{Simulating a directed opening tunnel.}
    \label{fig:scd-unbutton}
  \end{subfigure}
  \caption{Directed \emph{open-buttoned} self-closing door planarly simulates a parallel and an antiparallel directed \emph{open-tunneled} self-closing door.}
\end{figure}

We can use diodes to turn an opening button of any gadget into a directed tunnel (in either direction), as shown in \figref{fig:scd-unbutton}. We will use this implicitly in many of our constructions.

\subsection{Universality}
\label{sec:scd-universality}
In this section, we show that the self-closing door is universal, \emph{ignoring planarity}. We will recover planarity in the next section.

As a useful first step, we show how to make a self-closing door with many opening and self-closing tunnels:

\begin{lemma}\label{lem:scd-edge-duplicator}
  The self-closing door can simulate a version of the directed open-tunneled self-closing door with the tunnels duplicated arbitrarily many times.
\end{lemma}
\begin{proof}
  We use the simulation shown in \figref{fig:scd-edge-duplicator}. The figure shows $2$ opening tunnels and $2$ self-closing tunnels, but the same construction works with any number of tunnels.
  
  Suppose the agent enters the gadget through $A_i$. The only thing the agent can do is open door~$i$, go through door $0$, and self-close some door $j$ and leave through $B_j$. However, at this point, the only open door is door $i$, so $j = i$.
\end{proof}

\begin{figure}
  \centering
  \includegraphics[scale=.8]{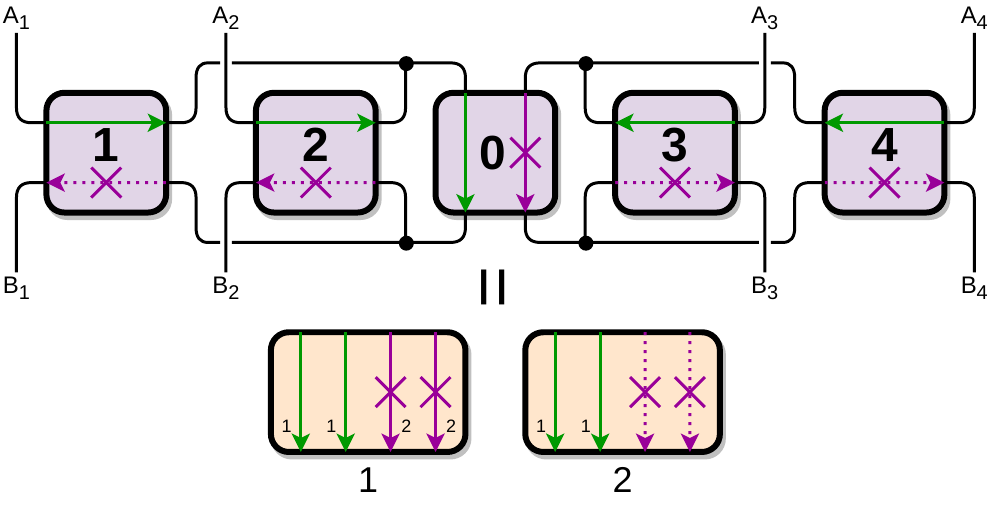}
  \caption{The directed open-tunneled self-closing door simulates a version of itself with the tunnels duplicated.}
  \label{fig:scd-edge-duplicator}
\end{figure}

\begin{theorem}\label{thm:scd-universal}
  Any self-closing door or symmetric self-closing door is universal, meaning it can simulate any gadget.
\end{theorem}

\begin{proof}
  Consider an arbitrary gadget $G$, with a set $L$ of locations, a set $S$ of states, an initial state $s_0$, and a set $T \subseteq S \times L \times S \times L$ of allowed transitions.

  For each state, we place a directed open-tunneled self-closing door $D_s$ with $|T|$ opening tunnels and $|T|$ self-closing tunnels as shown in \figref{fig:scd-universality-example}
  (implemented using Lemmas~\ref{lem:canonical-scd} and \ref{lem:scd-edge-duplicator}).
  The door $D_{s_0}$ corresponding to the starting state starts open, and all other doors start closed.

  For each transition $t = (s, l, s', l') \in T$, we add a path that starts at $l$, goes through a not-yet-used self-closing tunnel of door $D_s$, through a not-yet-used opening tunnel of door $D_{s'}$, and ends at $l'$.

  It is easy to see that this gadget behaves exactly like $G$.
\end{proof}

\begin{figure}
  \centering
  \includegraphics[scale=.8]{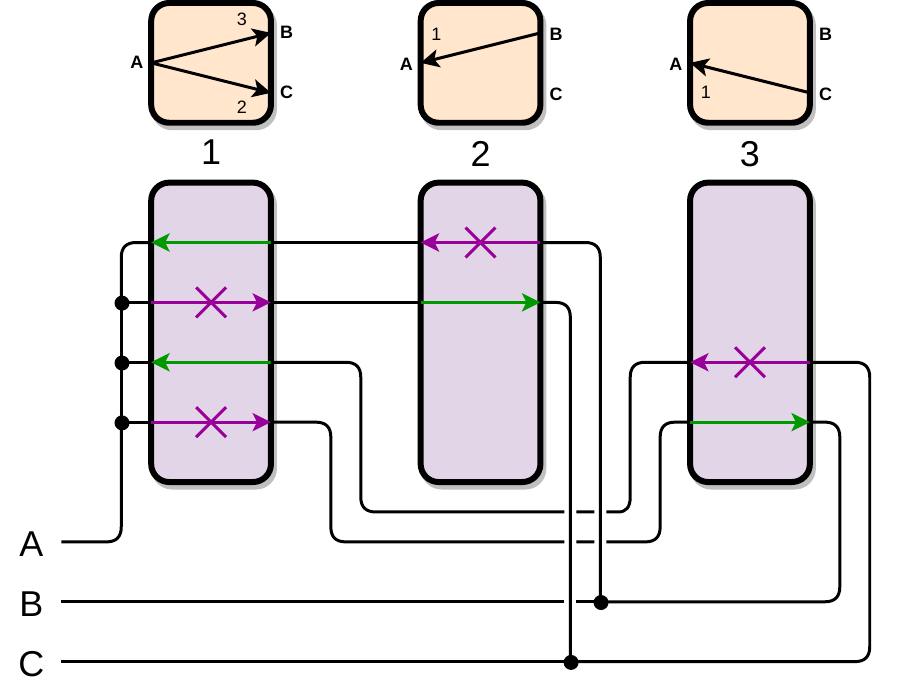}
  \caption{The directed open-tunneled self-closing door with duplicated tunnels simulates an arbitrary gadget.}
  \label{fig:scd-universality-example}
\end{figure}

\subsection{Planar Universality}
\label{sec:scd-planar-universality}

In this section, we show that the universality of \thmref{thm:scd-universal}
holds even in the planar setting.
It suffices to build a crossover gadget.
Our crossover was also given in \cite{pullingBlocks},
but we repeat it here for completeness.

We cannot duplicate tunnels as in \lemref{lem:scd-edge-duplicator}, as that construction is not planar; however, we can duplicate buttons:

\begin{lemma}\label{lem:scd-port-duplicator}
  The self-closing door simulates a directed open-buttoned self-closing door with two opening buttons on the same side of the self-closing tunnel.
\end{lemma}

\begin{proof}
  \figref{fig:scd-port-duplicator} shows the simulation.
  The agent enters $O_i$, opens door $i$, has a chance to open door $0$, and then must exit back through $O_i$, closing door $i$.
\end{proof}

We can then use diodes to turn some of these buttons into directed tunnels if necessary.

\begin{figure}
  \centering
  \includegraphics[scale=.8]{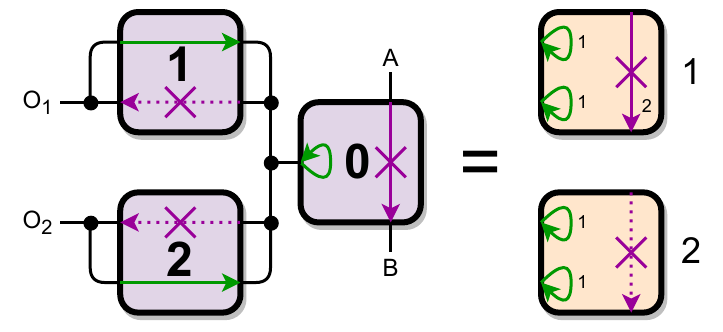}
  \caption{Self-closing doors planarly simulate a directed open-buttoned self-closing door with two opening buttons.}
  \label{fig:scd-port-duplicator}
\end{figure}

\begin{lemma}\label{lem:scd-crossover}
  The self-closing door planarly simulates a directed crossover.
\end{lemma}

\begin{proof}
  First we can simulate a pair of overlapping self-closing doors using the monstrous construction shown in \figref{fig:scd-monster}. If the agent enters from $O_1'$ or $O_2'$, they can open door $5$ or $6$, respectively, and then leave. We claim that the same is also true for $O_1$ and $O_2$.

  Suppose the agent enters from $O_1$. Then they open doors $2$, $3$, and $4$ and go through the diode. After this they have to choose between going up through door $3$ and going down through $4$. Assume they then traverse door $4$. If they then open door $6$, they would have to traverse door $3$, open $5$, and get stuck. So instead of opening door $6$, the agent traverses door $2$, ending up back at $O_1$ with no change except that door $3$ is open. However, this is completely useless, because the agent would be able to open $3$ anyway the next time they enter $O_1$ or $O_2$.
  
  So instead of going down through door $4$, the agent goes up through door $3$. The agent is then forced to go right and open door $5$, and traverse door $4$. If the agent opens door $6$, they will be stuck, so the agent traverses door $2$ instead and returns to $O_1$, leaving door $5$ open.
  
  Similarly, if the agent enters from $O_2$, the only useful thing they can do is open door $6$ and return to $O_2$. This means that $O_1$ and $O_1'$ act like opening buttons for $A_1 \to B_1$, and $O_2$ and $O_2'$ act like opening buttons for $A_2 \to B_2$.

  \begin{figure}
    \centering
    \includegraphics[scale=.8]{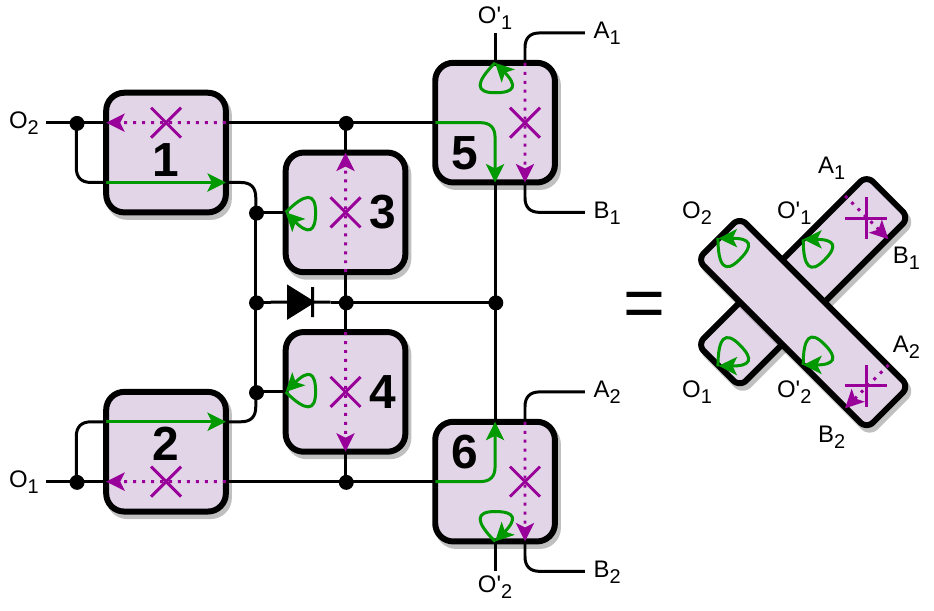}
    \caption{Directed self-closing door planarly simulates a pair of overlapping self-closing doors.}
    \label{fig:scd-monster}
  \end{figure}

  Next we connect two copies of this gadget as shown in \figref{fig:scd-dir-crossover}. Suppose the agent enters at $A_1$. They open door $1$. Now they have a choice of whether to open door $2$ or door $4$. If they open door $2$, then they open door $1$ again, traverse door $1$, and get stuck on door $3$. So instead they open doors $4$ and $3$, and traverse doors $1$ and $3$. At this point door $2$ is closed but door $4$ is open, so the agent must traverse door 4 and exit through $B_1$.

  \begin{figure}
    \centering
    \includegraphics[scale=.8]{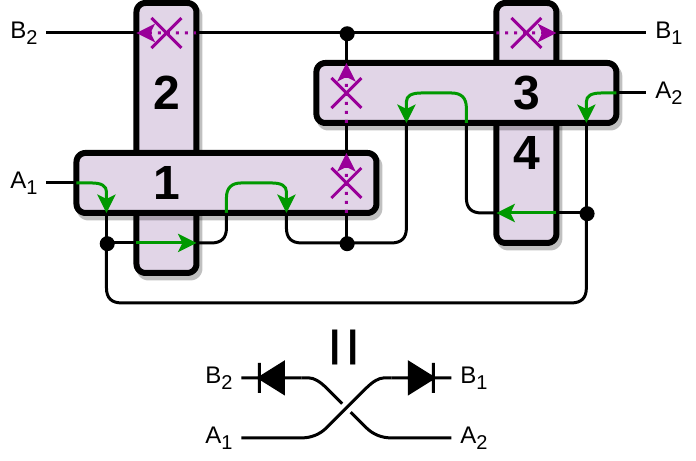}
    \caption{Overlapping directed self-closing door planarly simulates a directed crossover.}
    \label{fig:scd-dir-crossover}
  \end{figure}

  Finally, 4 copies of a directed crossover simulate an undirected one as shown in \figref{fig:dir-crossover}.
\end{proof}

\begin{figure}
  \centering
  \includegraphics[scale=.8]{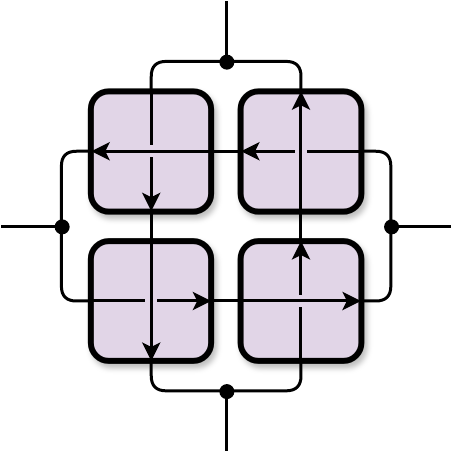}
  \caption{Directed crossover simulating an undirected crossover.}
  \label{fig:dir-crossover}
\end{figure}

\begin{theorem}\label{thm:planar-scd-universal}
  Any self-closing door or symmetric self-closing door is planarly universal, meaning it can planarly simulate any gadget.
\end{theorem}

\begin{proof}
  This follows from Lemmas~\ref{lem:canonical-scd} and~\ref{lem:scd-crossover} together with \thmref{thm:scd-universal}.
\end{proof}

\begin{corollary}\label{thm:planar-scd-pspace}
  Planar reachability with any self-closing door or symmetric self-closing door is \PSPACE-complete.
\end{corollary}

\begin{proof}
  1-player reachability is known to be in \PSPACE in general,
  and \PSPACE-hard for at least one set of gadgets
  \cite{gadgets,gadgets2}.
\end{proof}

\section{Open--Close Doors}
\label{sec:doors}

In this section, we show that any open--close door is planarly universal, and therefore planar reachability with any open--close door is \PSPACE-complete.
First we define the gadget formally.

An \bemph{open--close door} is a 2-state gadget with an \emph{opening} button or tunnel, a \emph{traverse} tunnel, and a \emph{closing} tunnel.
An open--close door is \bemph{open-tunneled} if it has an opening tunnel and \bemph{open-buttoned} if it has an opening button.
The opening button/tunnel opens the traverse tunnel, and the closing tunnel closes the traverse tunnel.
\figref{fig:different-doors} shows three examples of open--close doors.

Similar to self-closing doors,
each of the (two or three) tunnels in an open--close door can be either directed or undirected.
We call an open--close door \bemph{directed} if all of its tunnels are directed, \bemph{undirected} if all of its tunnels are undirected, and \bemph{mixed} if it has both directed and undirected tunnels.

Throughout this paper, the opening button/tunnel will be colored green, the traverse tunnel will be colored blue, and the closing tunnel will be colored red.
A dotted line indicates a closed traverse tunnel and a solid line indicates an open traverse tunnel.

Without planarity, we can wire together the close and traverse tunnels of an open--close door to simulate some self-closing door.
By \thmref{thm:scd-universal}, every open--close door is universal,
and by Corollary~\ref{thm:planar-scd-pspace}, reachability with every open--close door is \PSPACE-complete.
Nonplanar \PSPACE-hardness also follows from \cite[Metatheorem~3]{lemmings}.

Thus the goal of this section is to show that any open--close door is \emph{planarly} universal.
We will do so by showing that each open--close door simulates some self-closing door, some symmetric self-closing door, or a crossover.
First, in Sections~\ref{sec:planar-undirected} and~\ref{sec:planar-mixed}, we show planar universality for undirected and mixed open--close doors, respectively.
Next, in Section~\ref{sec:planar-crossing}, we show planar universality for all directed open--close doors with at least one pair of crossing tunnels.
Finally, in Section~\ref{sec:planar-noncrossing}, we show planar universality for directed open--close doors without internal crossings, where we have twelve different cases shown in \figref{fig:dir-door-cases}.

\subsection{Undirected Open--Close Doors}
\label{sec:planar-undirected}

For undirected open--close doors, the simple simulation of a self-closing door described above can be made planar:

\begin{lemma}\label{lem:planar-undirected}An undirected open--close door can planarly simulate a self-closing door.
\end{lemma}
\begin{proof}
  If the door is open-tunneled, we block off one of the ends of the opening tunnel. The agent can still enter from the other end of the opening tunnel, go through the opening tunnel twice, and leave. This effectively turns an undirected opening tunnel into an opening button. So without loss of generality, our gadget is open-buttoned.

  There are only four distinct ways that the 5 locations can be ordered, as shown in \figref{fig:undir-door-cases}.
  (If closing and traverse tunnels cross, then the opening button can be in any of the four regions, which are all symmetric to each other; if they do not cross, then only 2 regions are symmetric, so we get 3 cases.)
  In each case, we can wire the traverse and closing tunnels in series without blocking the opening button, as shown in \figref{fig:undir-door-cases}.
  This simulates an open-buttoned self-closing door:
  the player can open the gadget by going to the opening button;
  if the gadget is open, then the player can go through the traverse tunnel and then the closing tunnel,
  but cannot go the other way; and
  if the gadget is closed, then the player cannot go either way through the traverse tunnel or the closing tunnel.
\end{proof}

\begin{figure}
  \centering
  \includegraphics[scale=1]{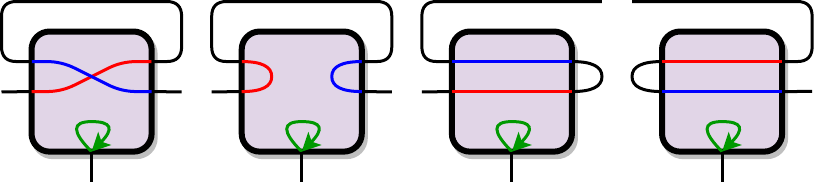}
  \caption{Any undirected open-buttoned open--close door simulates an open-buttoned self-closing door.}
  \label{fig:undir-door-cases}
\end{figure}

\subsection{Mixed Open--Close Doors}
\label{sec:planar-mixed}

Next we reduce the mixed case to the directed case, which we will analyze in the next two subsections.

\begin{lemma}\label{lem:planar-mixed}Any mixed open--close door can planarly simulate some directed open--close door.
\end{lemma}
\begin{proof}
  Consider an arbitrary mixed open--close door $M$.
  Because $M$ is mixed, it has a directed tunnel.
  No tunnel changes its own traversability when crossed, so this tunnel simulates a diode. We wire each
  undirected tunnel of $M$ through diodes at each end pointing in the same direction. This simulates a directed open--close door.
\end{proof}

\subsection{Directed Open--Close Doors with Internal Crossings}
\label{sec:planar-crossing}

It remains to analyze directed open--close doors.
First we consider a case where we can build a crossover.

An \defn{internal crossing} of a planar gadget
is a pair of tunnels $(a,b),(c,d)$ that cross each other
when drawn interior to the gadget,
i.e., where locations $a,c,b,d$ appear in this cyclic order.
For example, in \figref{fig:undir-door-cases}, Case 1 has an internal crossing.

\begin{lemma}\label{lem:planar-crossing} Any directed open--close door with an internal crossing can planarly simulate a directed crossover.
\end{lemma}
\begin{proof} If the opening tunnel crosses the closing tunnel, then we have a crossover because these tunnels are always open.
If the opening tunnel crosses the traverse tunnel, then we start the door open and have a crossover because neither
tunnel closes itself or the other.
Otherwise, the traverse tunnel crosses the closing tunnel and the opening button/tunnel can simulate an opening button.
Then we have four cases, as shown in
\figref{fig:traverse-closing-cases}.

\begin{figure}
  \centering
  \includegraphics[scale=.8]{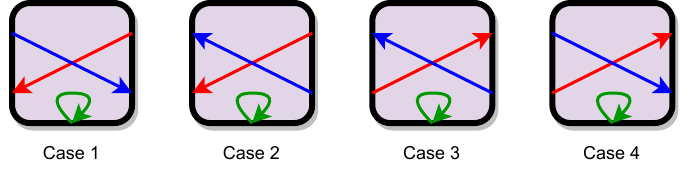}
  \caption{The four cases where the traverse tunnel crosses the closing tunnel but the opening button/tunnel does not cross either and can thus simulate a button.}
  \label{fig:traverse-closing-cases}
\end{figure}

In Cases 1, 2, and 3, we can simulate a crossover by connecting the opening button to either the input of the traverse tunnel or the output of the closing tunnel to ensure that the traverse tunnel is open when we need to use it,
as shown in \figref{fig:traverse-closing-cases-1234}.

\begin{figure}
  \centering
  \includegraphics[scale=.7]{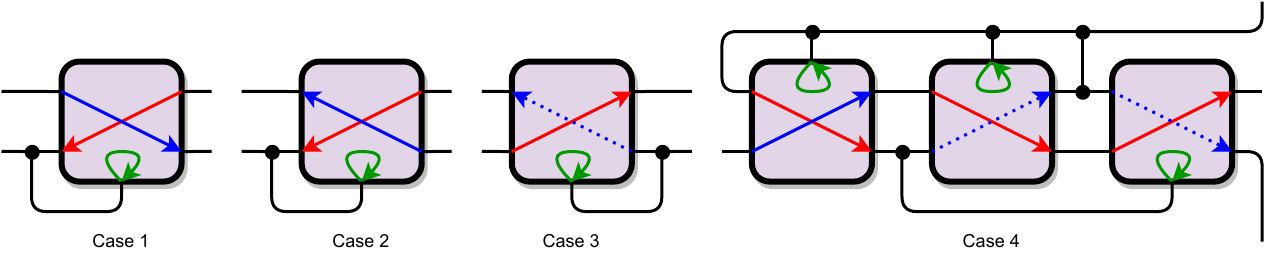}
  \caption{All four cases of the traverse tunnel crossing the closing tunnel can each simulate a crossover.}
  \label{fig:traverse-closing-cases-1234}
\end{figure}

Case 4 is trickier, however, because the opening button is adjacent to neither the input of the traverse tunnel nor the output of the closing tunnel.
In this case, we use three copies of the gadget, as shown on the right in \figref{fig:traverse-closing-cases-1234}.
The left-to-right path of this crossover involves going through all three doors, which is allowed as long as the left door stays open.
To take the top-to-bottom path,
the player opens the middle door, closes the left door, opens the right door, traverses the middle door, reopens
the left door (to keep the left-to-right path open), and traverses the right door. The player can leave partway through this
traversal, but this does nothing useful. So all doors with internal crossings can simulate crossovers.
\end{proof}

Once we have a crossover, we can use it to connect the traverse and closing tunnels in series to make a self-closing door, which we know is universal by \thmref{thm:scd-universal}.

\subsection{Directed Open--Close Doors without Internal Crossings}
\label{sec:planar-noncrossing}

Finally, this section solves the remaining (hardest) situation:

\begin{lemma}\label{lem:planar-noncrossing}
  Any directed open--close door without internal crossings can planarly simulate a self-closing door.
\end{lemma}

If the open--close door is open-tunneled and the ends of the opening tunnel are adjacent,
then we can connect them together to simulate an opening button, reducing the number of cases to consider.
This leaves twelve cases, shown in \figref{fig:dir-door-cases}.
We name these cases based on the cyclic order of locations, with
entrances having uppercase letters and exits having lowercase letters.

\begin{figure}
  \centering
  \includegraphics[scale=1]{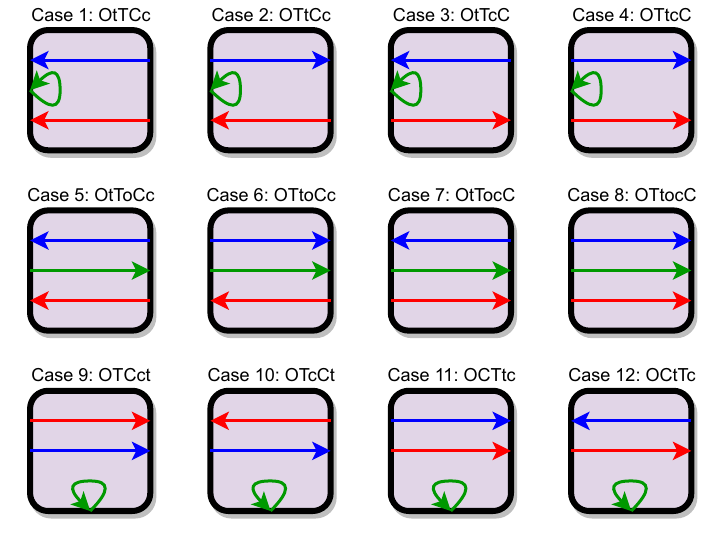}
  \caption{The twelve cases of a planar directed open--close door without internal crossings. Opening tunnels with adjacent locations are merged into an opening button.}
  \label{fig:dir-door-cases}
\end{figure}

We consider each case one by one in the following subsections, ordered roughly by the complexity of the proof, from simplest to most complicated.
The cases will not be considered in numeric order as defined in \figref{fig:dir-door-cases}.
We show that each case simulates a self-closing door or one of the cases considered earlier.

All of these cases except Case 8: OTtocC were solved using a bottom-up computer search from \cite{bosboom-thesis}. The basic idea is to repeatedly connect two existing gadgets together to make a new gadget, and hopefully eventually reach a known-universal gadget (like a self-closing door or a directed tripwire--lock). We also had hand-crafted proofs for many of these cases, but they were matched or surpassed by the computer search. We try to present these computer-generated results in the most humanly sensible way whenever possible.

The bottom-up computer search failed to solve Case 8: OTtocC. The conference version of this paper~\cite{doors-fun2020} listed this case as open, but we have since solved it with a dedicated computer search algorithm described in Section~\ref{sec:door-8}.

\subsubsection{Case 2: OTtCc, Case 10: OTcCt, and Case 12: OCtTc}
\label{sec:door-2-10-12}

In all these open--close doors, the traverse tunnel output is adjacent to the
closing tunnel input. Thus, we can simulate a self-closing door by wiring the traverse tunnel output to the
closing tunnel input.

\subsubsection{Case 4: OTtcC and Case 9: OTCct}
\label{sec:door-4-9}

These open--close doors can simulate a directed open-buttoned self-closing door, as shown in \figref{fig:dir-door-case-49}. Note that the two figures differ only in the planar embedding, so the proof is exactly the same for both of them.

If the agent
enters from $O$, they can open door 1, and then they are forced to close door 2. Continuing this loop does nothing, so the agent
then returns to $O$. Now door 1 is open and door 2 is closed. If the agent then enters from $A$, then they are forced to traverse
door 1. They can then open door 2 and then they are forced to close door 1. Continuing the loop does nothing, so the agent has no
other option but to traverse door 2 to $B$. The agent could not have taken this path initially since door 1 was closed, and
they cannot take it again without visiting $O$ because they closed door 1.

\begin{figure}
  \centering
  \begin{subfigure}{0.49\textwidth}
    \centering
    \includegraphics[scale=.8]{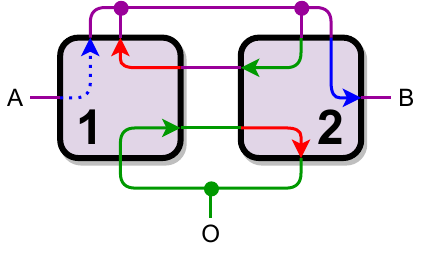}
    \caption{Case 4: OTtcC simulates a self-closing door.}
    \label{fig:dir-door-case-4}
  \end{subfigure}
  \begin{subfigure}{0.49\textwidth}
    \centering
    \includegraphics[scale=.8]{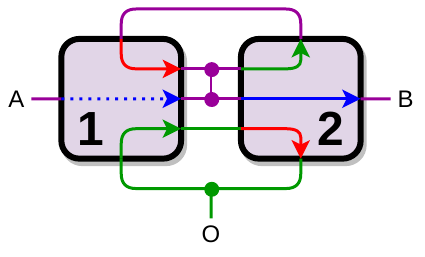}
    \caption{Case 9: OTCct simulates a self-closing door.}
    \label{fig:dir-door-case-9}
  \end{subfigure}
  \caption{Case 4: OTtcC and Case 9: OTCct simulate a self-closing door. The simulations start in the closed states. Locations and gadgets are labeled.}
  \label{fig:dir-door-case-49}
\end{figure}

\subsubsection{Case 6: OTtoCc Door}
\label{sec:door-6}

This open--close door can simulate a directed open-buttoned self-closing door, as shown in \figref{fig:dir-door-case-6}. If the agent
enters from $O$, they are forced to close door 3. If the agent then traverses door 2, they are forced to open door 3 and return
to $O$, accomplishing nothing. So the agent has no other option but to close door 1. If the agent tries to open door 2, they get stuck,
so they instead open door 1. Continuing the loop involving door 1 does nothing, so the agent then traverses door 2, opens door 3, and returns
to $O$. Now door 1 is open. If the agent enters from $A$, then they are forced to close door 2, traverse door 1, and close door 1.
Reopening door 1 puts the agent back into the situation of being forced to close door 1, so the agent instead opens door 2 and traverses door 3
to $B$. The agent could not have taken this path initially since door 1 was closed, and they cannot take it again
without visiting $O$ because they closed door 1.

\subsubsection{Case 7: OtTocC Door}
\label{sec:door-7}

This open--close door can simulate a directed open-buttoned self-closing door, as shown in \figref{fig:dir-door-case-7}.
If the agent enters from $O$, they must open door 3, then close door 2. If the agent then closes door 1, they get stuck because
door 2 is closed. The agent can traverse door 3 and leave via $O$, but they can also open and then traverse door 1 and then do the
same thing, which is advantageous. So the agent opens and traverses door 1, then traverses door 3 to $O$. Now door 3 is open,
door 2 is closed, and door 1 is open. If the agent enters from $A$, they must close door 3, then open door 2, then
traverse door 1. Opening door 1 and then traversing it is a no-op, and door 3 is closed, so the agent closes door 1 and then must
traverse door 2 to $B$. This leaves door 1 closed, door 2 open, and door 3 closed. The agent could not have taken this path initially
because door 1 was closed, and cannot take it again without visiting $O$ first for the same reason.

\begin{figure}
  \centering
  \begin{subfigure}{0.49\textwidth}
    \centering
    \includegraphics[scale=.8]{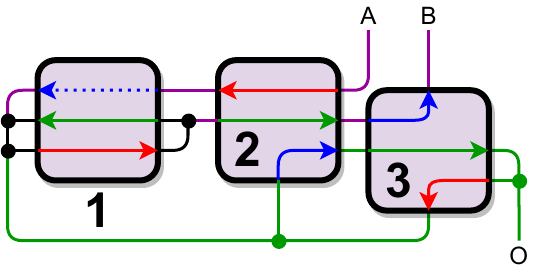}
    \caption{Case 6: OTtoCc simulates a self-closing door.}
    \label{fig:dir-door-case-6}
  \end{subfigure}
  \begin{subfigure}{0.49\textwidth}
    \centering
    \includegraphics[scale=.8]{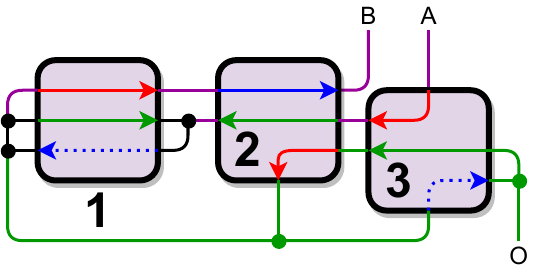}
    \caption{Case 7: OtTocC simulates a self-closing door.}
    \label{fig:dir-door-case-7}
  \end{subfigure}
  \caption{Case 6: OTtoCc and Case 7: OtTocC simulate a self-closing door. The simulations start in the closed states. Locations and gadgets are labeled.}
\end{figure}

\subsubsection{Case 3: OtTcC Door}
\label{sec:door-3}

This open--close door can simulate a directed open-buttoned self-closing door, as shown in \figref{fig:dir-door-case-3}. If the agent
enters from $O$ (the opening button), they can open doors 2 and 3. If they then leave, they have accomplished nothing because door 2 was
already open, and door 3 can be opened from $O$ anyway and cannot be traversed from $A$ or $B$ as we will see later.
So they close door 2 instead. Then they can open door 1 and they are forced to traverse door 3. The agent can then reopen
door 2 and return to $O$.
Now all the doors are open. If the agent then enters from $A$, then they are forced to traverse door 1 and close door 3. They can
then open door 1 (useless), and then they are forced to traverse door 2 and close door 1, leading to $T_1$.
The agent could not have taken this path initially
because door 1 was closed, and they cannot take it again without visiting $O$ because they just closed door 1.

\begin{figure}
  \centering
  \includegraphics[scale=.8]{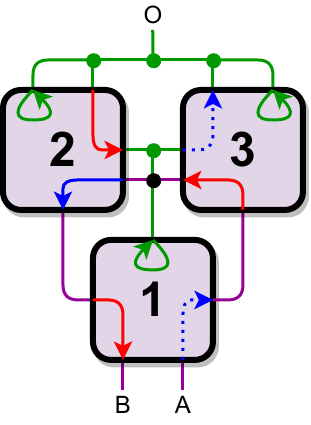}
  \caption{Case 3: OtTcC simulates a self-closing door. The simulation starts in the closed state. Locations and gadgets are labeled.}
  \label{fig:dir-door-case-3}
\end{figure}

\subsubsection{Case 1: OtTCc Door}
\label{sec:door-1}

In this case, we end up constructing a gadget we call the ``directed tripwire--lock''.
Recall that a \defn{tripwire--lock} \cite{gadgets} is a 2-state 2-tunnel gadget with an undirected tunnel that is traversable in exactly 1 state and an undirected tunnel that toggles the state of the gadget.
The \bemph{directed tripwire--lock} is similar except that its tunnels are directed. It turns out this gadget is also universal:

\begin{lemma}
  The OtTCc open--close door can simulate the directed tripwire--lock.
\end{lemma}
\begin{proof}
  This open--close door can simulate the parallel directed tripwire--lock, as shown in \figref{fig:dir-door-case-1}. The lock is simply the traverse tunnel on door 1. In the two simulated states we will either have doors 1 and 3 open or door 2 open. If door 2 is open, when traversing the tripwire tunnel we can go through the traverse tunnel allowing us to open doors 1 and 4 before closing door 2. With door 4 now open, we can go through its traverse tunnel opening door 3, and then closing door 4 on the way out. This leaves us with doors 1 and 3 open. Going through the tripwire tunnel again closes door 1 but allows us to go through the traverse tunnel of door 3 and open door 2. Doors 3 and 4 are then closed on the way out. There are states where we could fail to open all of these doors while traversing the tripwire tunnel, but this will leave the gadget with strictly less traversability and thus the agent will never want to take such a path.
\end{proof}
 
\begin{figure}
  \centering
  \includegraphics[scale=.8]{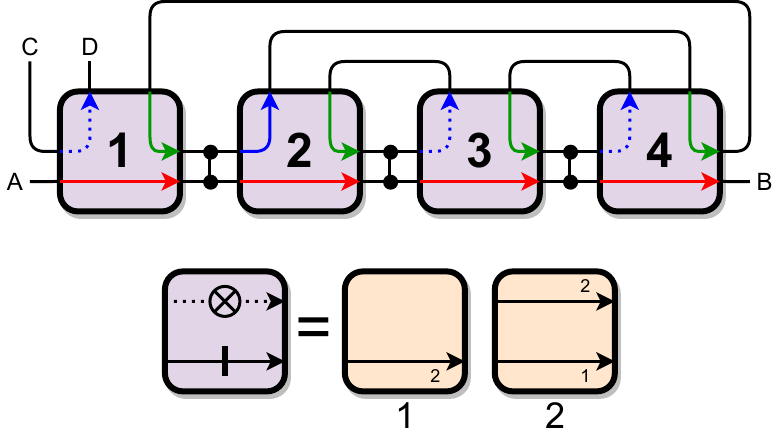}
  \caption{Top: Case 1: OtTCc simulates the parallel directed tripwire--lock.
    Bottom: the notation and the state diagram of the parallel directed tripwire--lock.}
  \label{fig:dir-door-case-1}
\end{figure}

\begin{lemma}\label{lem:dir-tripwire-lock-scd}
  The parallel directed tripwire--lock simulates a self-closing door.
\end{lemma}
\begin{proof}
The parallel directed tripwire--lock can simulate the antiparallel directed tripwire--lock,
as in \figref{fig:dir-door-case-1-antiparallel-DWL}. Crossing the tripwire in gadget 2 forces the agent to cross the tripwire
in gadget 1, so exactly 1 of the locks of gadgets 1 and 2 is locked. Similarly, exactly 1 of the locks of gadgets 3 and 4
is locked. Every traversal from $A$ to $B$ has to traverse the tripwire in gadget 5 exactly once, so the lock in gadget 5 is unlocked after an even number of $A \to B$ traversals and locked after an odd number of $A \to B$ traversals. Since the locks of gadgets 1 and 2 are anti-correlated, if the agent
wants to unlock the lock in gadget 1, it must afterward cross the lock in gadget 3. But said lock must be unlocked by
going in a loop through the lock in gadget 5, which is unlocked only after an even number of $A \to B$ traversals.
So during an even-indexed (second, fourth, etc.) traversal, the lock in gadget 1 must be locked before the agent
can leave. During an odd-indexed traversal, the agent can take the loop through the lock in gadget 5, unlock the lock in gadget 1,
take the loop again to unlock the exit, and exit. So the top path behaves like a directed tripwire for the lock in gadget 1,
which is the bottom path.

\begin{figure}
  \centering
  \includegraphics[scale=.8]{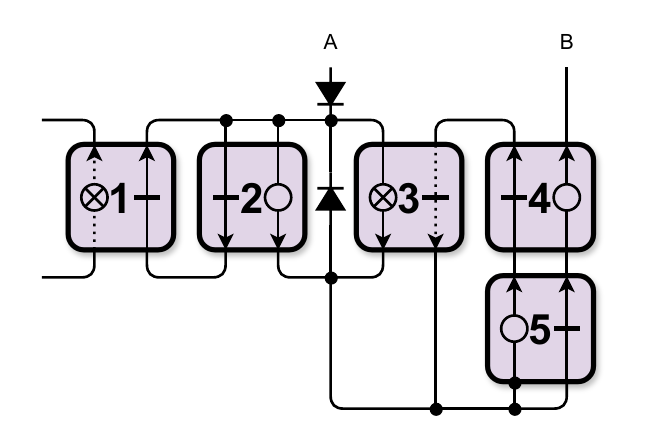}
  \caption{Parallel directed tripwire--lock simulates an antiparallel directed tripwire--lock.}
  \label{fig:dir-door-case-1-antiparallel-DWL}
\end{figure}

The antiparallel directed tripwire--lock simulates a directed tripwire--lock--tripwire (\figref{fig:dir-door-case-1-DWLW}, left). In this gadget, going through the top or bottom pairs of tripwires flips which of the middle locks is closed. If the top and bottom sets of locked tunnels are paired, then they block the middle pathway; however, if the sets are opposite, then the connection in the middle can be used to traverse them.

\begin{figure}
  \centering
  \includegraphics[scale=.8]{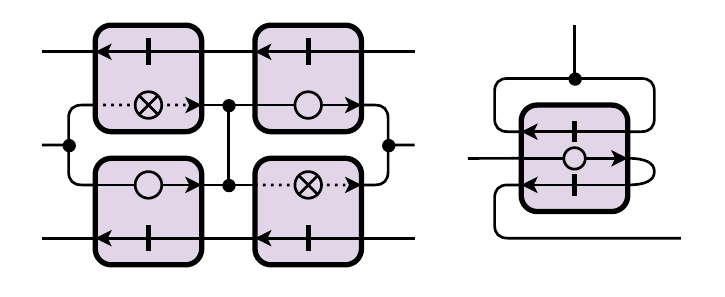}
  \caption{Left: antiparallel directed tripwire--lock simulates a directed tripwire--lock--tripwire; inspired by Figure 10 of \cite{gadgets}. Right: directed tripwire--lock--tripwire simulates a self-closing door.}
  \label{fig:dir-door-case-1-DWLW}
\end{figure}

This gadget in turn simulates a self-closing door (\figref{fig:dir-door-case-1-DWLW}, right). Traversing the horizontal tunnel closes the gadget, and it can be reopened by visiting the top location.
\end{proof}

\subsubsection{Case 5: OtToCc Door}
\label{sec:door-5}

This open--close door can simulate the Case 6: OTtoCc door, which has been covered in Section~\ref{sec:door-6},
by effectively flipping the traverse tunnel (\figref{fig:dir-door-case-5}). Door 1 is the gadget that we flip the traverse tunnel of.
If the agent enters from $A$, they must open door 2, then close door 2. If door 1 is open and the agent then traverses it, they are
back to a previous position with nothing changed. Instead, the agent opens door 3. If the agent then closes door 3, they get stuck because
door 2 is closed. So they must close door 2 (again) or traverse
door 3. These actions lead to the same situation. If the agent opens door 3 (again), they are back to the same situation that occurred after
opening door 3 the first time. If door 1 is open, the agent then traverses door 1. Then they must open door 2. Closing door 2 leads
to a previous situation, so the agent then traverses door~3. If the agent then traverses door 1 (again), they must open door 2 (again),
leading to a previous situation. So they instead open door~3. Closing door 2 and traversing door 3 lead to different previous situations,
so the agent then closes door 3, and then is forced to traverse door 2 to $B$, leaving all the doors unchanged. If door 1 is not open,
then the agent is unable to leave.

\begin{figure}
  \centering
  \includegraphics[scale=.8]{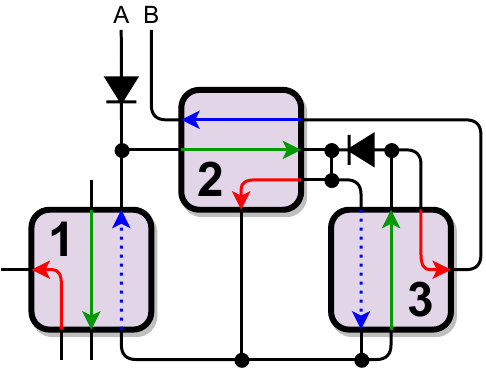}
  \caption{Case 5: OtToCc simulates Case 6: OTtoCc. The traverse tunnel of the leftmost gadget is effectively flipped.}
  \label{fig:dir-door-case-5}
\end{figure}

\subsubsection{Case 11: OCTtc Door}
\label{sec:door-11}

This open--close door can simulate the Case 12: OCtTc door, which has been covered in Section~\ref{sec:door-2-10-12},
by effectively flipping the traverse tunnel (\figref{fig:dir-door-case-11}). Door 1 is the gadget that we flip the traverse tunnel of.
If the agent enters from $A$, then they must traverse the left and middle diodes. The agent can then open doors 2 and 4 but must pick
one to close. If they close door 2, they get stuck. So the agent closes door 4. Going to open doors 2 and 4 again leads to a previous
situation, so the agent instead traverses door 2. Traversing door 1 (if it is open) leads to a previous situation,
and traversing door 5 leads to being stuck (since the agent has previously closed door~4).
The agent traverses the right diode, and can open doors 3 and 5 but must pick one to close. Closing door 3 leads to a previous
situation, so the agent then closes door 5, then must traverse door 3. Going to open doors 3 and 5 again leads to a previous situation. If
door 1 is open, then the agent traverses door 1. Traversing door 2 leads to a previous situation, so the agent then opens doors 2 and 4.
Closing door 4 leads to a previous situation, so the agent closes door 2, then must traverse door 3. Traversing door 1 leads to a
previous situation, so the agent instead opens doors 3 and 5. Closing door 5 leads to a previous situation, so the agent then
closes door 3. Traversing door 1 or going to open doors 3 and 5 both lead to previous situations, so the agent then traverses doors 5 and 3, leaving via $B$ and leaving all the doors unchanged. If door 1 is not open, however, then the agent
cannot leave because door 4 is closed and the only way to open it is through door 1's traverse tunnel.

\begin{figure}
  \centering
  \includegraphics[scale=.8]{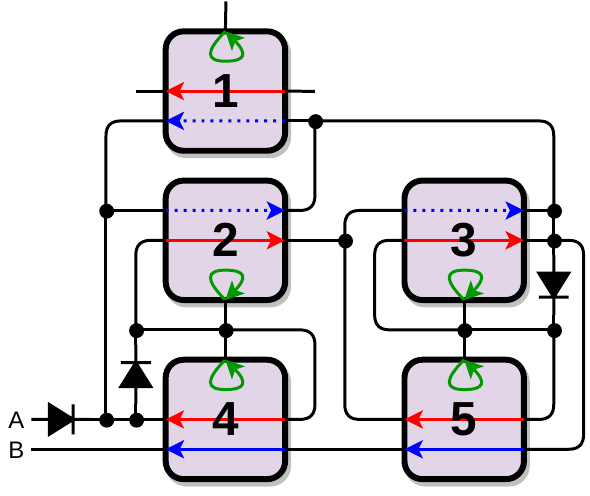}
  \caption{Case 11: OCTtc simulates Case 12: OCtTc. The traverse tunnel of the topmost gadget is effectively flipped.}
  \label{fig:dir-door-case-11}
\end{figure}

\subsubsection{Case 8: OTtocC Door}
\label{sec:door-8}

This open--close door can simulate the Case 7: OtTocC door, which has been covered in Section~\ref{sec:door-7}, by effectively flipping the traverse tunnel (\figref{fig:dir-door-case-8}). A straightforward but tedious depth-first search through the configuration space shows that the only way to go from $A$ to $B$ is to go straight at every intersection; turning at any intersection gets you stuck. This traversal preserves the states of all gadgets, but is only possible if door 10 is open.

\begin{figure}
  \centering
  \includegraphics[scale=.8]{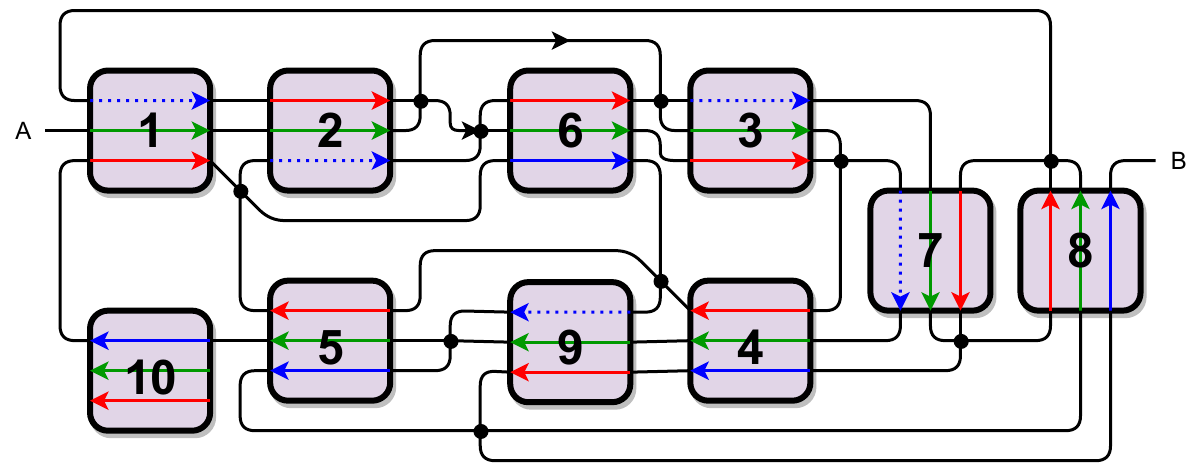}
  \caption[Case 8: OTtocC simulates Case 7: OtTocC.]
    {Case 8: OTtocC simulates Case 7: OtTocC.
    The traverse tunnel of gadget 10 is effectively flipped.
    The traversal sequence simulating a Case-7 door traverse (from $A$ to $B$) is the following: open~1, open~2, open~3, close~4, close~5, traverse~2, close~6, traverse~3, open~7, close~8, traverse~1, close~2, open~6, close~3, traverse~7, open~4, open~9, open~5, traverse~10, close~1, traverse~6, traverse~9, traverse~5, open~8, close~7, traverse~4, close~9, traverse~8.
    See an animation at \url{https://erikdemaine.org/gadgets/evildoor/}.}
  \label{fig:dir-door-case-8}
\end{figure}

Next we sketch the algorithm we used to find this proof.

We start with a directed closed loop in the plane, with some self-crossings. We want to add some doors so that the only infinite periodic walk through the gadgets is around this loop, and turning at any crossing eventually gets you stuck. We look for triples of edges $c, o, t$ such that
\begin{itemize}
\item the edges appear in cyclic order $cot$ in the loop (as opposed to $toc$);
\item $c$ and $o$ share a face $f_{co}$ and are directed in opposite directions around $f_{co}$; and
\item $o$ and $t$ share a different face $f_{ot}$ and are directed in opposite directions around $f_{ot}$.
\end{itemize}

For every such triple of edges, we add a door with $c$ as the closing tunnel, $o$ as the opening tunnel, and $t$ as the traverse tunnel. This is enough to make sure that the only periodic walk is the desired one. However, this creates a huge number of doors (around $200$), many of which overlap each other (in the sense of \figref{fig:scd-dir-crossover}). We choose a random door (with doors with lots of overlaps weighted higher), and try removing it. If this introduces any undesired periodic walks, we backtrack. With a complicated enough starting loop and with some luck, we end up with a set of doors that do not overlap.

After this, we cut the loop somewhere ($A$ and $B$ in the figure) and insert the final door (door $10$ in the figure). This is not guaranteed to work, because, for example, you might be able to go from $A$ to $B$ a finite number of times even when the gadget is supposed to be closed (we only know you cannot do it infinitely). But in practice it usually works.

We ran this algorithm many times for many different starting loops, and the smallest solution we found was \figref{fig:dir-door-case-8}.

\subsection{Universality}

\begin{theorem}\label{thm:planar-door-universal} Any open--close door is planarly universal, i.e., it can planarly simulate any gadget.
\end{theorem}
\begin{proof}
This follows from Lemmas~\ref{lem:planar-undirected}, \ref{lem:planar-mixed}, \ref{lem:planar-crossing}, and~\ref{lem:planar-noncrossing},
as those cover all the cases.
\end{proof}

\begin{corollary}\label{thm:planar-door-pspace}
  Planar reachability with any open--close door is \PSPACE-complete.
\end{corollary}

\begin{proof}
  Analogous to Corollary~\ref{thm:planar-scd-pspace}.
\end{proof}

\section{Applications}
\label{sec:applications}
In this section, we use our results about the complexity of
planar reachability with door gadgets to prove \PSPACE-hardness
of several video games:

\begin{enumerate}
\item Lemmings, Donkey Kong Country 1, 2, and 3,
  The Legend of Zelda: A Link to the Past, and Super Mario Bros.\
  were previously proved \PSPACE-hard, but the proofs simplify
  because they no longer require crossover gadgets
  (Section~\ref{sec:simplifications}).
\item \emph{Sokobond} is a 2D block pushing game
  where the blocks are able to fuse into polyominoes,
  which we newly prove \PSPACE-complete (Section~\ref{sec:sokobond}).
\item Eight new Mario games are newly \PSPACE-complete:
  \begin{enumerate}
  \item Super Mario 64 and Super Mario 64 DS (Section~\ref{sec:sm64-sm64ds}).
  \item Super Mario Sunshine (Section~\ref{sec:super-mario-sunshine}).
  \item Super Mario Galaxy (Section~\ref{sec:super-mario-galaxy}).
  \item Super Mario Galaxy 2 (Section~\ref{sec:super-mario-galaxy-2}).
  \item Super Mario 3D Land and Super Mario 3D World (Section~\ref{sec:super-mario-3d-land-world}).
  \item Super Mario Odyssey (Section~\ref{sec:super-mario-odyssey}).
  \item Captain Toad:\ Treasure Tracker,
    and the associated Captain Toad levels in Super Mario 3D World
    (Section~\ref{sec:captain-toad}).
  \end{enumerate}
  All eight are 3D platformers in which the player controls the agent (either Mario or Toad) to collect resources or reach target locations while avoiding or defeating enemies and environmental hazards.
  The player's actions typically include making the agent jump and walk in an approximately continuous environment, but Toad is unable to jump.
  The agent also has health and can take damage, which can cause the player to lose the game.
  %One game, Captain Toad:\ Treasure Tracker, is a 3D puzzle platformer and is mechanically similar, except that the agent Toad is unable to jump.
  Each game's section gives the necessary details on any additional mechanics.
\end{enumerate}

\subsection{Simplifications}
\label{sec:simplifications}

Our planar open--close door results in particular simplify prior uses of a door framework.
%to no longer need crossover gadgets.
The Lemmings door \cite[Figure~4]{lemmings} has an internal crossing,
so \lemref{lem:planar-crossing} applies.
The Donkey Kong Country 1, 2, and 3 doors \cite[Figures~21--23]{nintendoor}
are the Case 10: OTcCt door, Case 4: OTtcC door, and an internal crossing door, respectively,
so Lemmas~\ref{lem:planar-noncrossing} and~\ref{lem:planar-crossing} apply.
The Legend of Zelda: A Link to the Past door \cite[Figure~30]{nintendoor}
has an internal crossing, so \lemref{lem:planar-crossing} applies.
The Super Mario Bros.\ door \cite[Figure~6]{demaine2016super}
is the Case 4: OTtcC door, so \lemref{lem:planar-noncrossing} applies.
Therefore all of the crossover gadgets in these reductions
\cite[Figure~2(e)]{lemmings},
\cite[Figure~20]{nintendoor},
\cite[Figure~28]{nintendoor},
\cite[Figure~5]{demaine2016super}
are not in fact needed to prove \PSPACE-hardness of these games.

\subsection{Sokobond}
\label{sec:sokobond}
Sokobond \cite{Sokobond} is a 2D block pushing game where the blocks are atoms/molecules. Movement is discrete along a square grid. The player starts as a single atom.
Each atom except He has some number of free
electrons (H has 1, O has 2, N has 3, C has 4). When two atoms that both have free electrons are adjacent, they both lose
a free electron and bond into a molecule. Molecules are rigid, so pushing an atom in a molecule results in the entire molecule
moving. Atoms/molecules can also push each other.

Sokobond with He atoms is trivially NP-hard as it includes \textsc{Push}\text{-}$*$ \cite{PushingBlocks_CGTA}.
We show \PSPACE-hardness even without He atoms:

\begin{theorem} Completing a level in Sokobond with H and O atoms is \PSPACE-hard.
\end{theorem}
\begin{proof}
We reduce from 1-player planar motion planning with open--close doors and use \thmref{thm:planar-door-pspace}.
 
Let the player start as an H atom trying to reach another H atom.
We can simulate an open--close door as shown in \figref{fig:sokobond-door}. To open the door,
the player pushes down on the big molecule. The player can go through the traverse tunnel if and only if the molecule is down.
When going through the closing tunnel, the player is forced to push up on the molecule, closing the traverse tunnel. The
molecule used to simulate an open--close door has no free electrons, so the level can be completed if and only if the player can reach the other H atom.
\end{proof}

%\begin{figure}
%  \centering
%  \includegraphics[scale=0.6]{sokobond/door}
%  \caption{Simulation of a door in Sokobond. The opening button is at the bottom left. The traverse tunnel is undirected
%  and runs between the top left and the top right. The closing tunnel is undirected and runs between the middle right and the
%  bottom right.}
%  \label{fig:sokobond-door}
%\end{figure}

\subsection{Super Mario 64/Super Mario 64 DS}
\label{sec:sm64-sm64ds}

\xxx{Would be awesome but not necessary to have screenshots of the relevant items/mechanics in the Mario games.}

Super Mario 64 is a 3D Mario game for the Nintendo 64
where Mario collects Stars from courses inside paintings to save the princess,
who is trapped behind a painting.
Super Mario 64 DS is a remake of Super Mario 64 for the Nintendo DS
(still in 3D), featuring the same courses as in Super Mario 64 plus new courses,
as well as the ability to play as characters other than Mario.
In this reduction, we will primarily make use of quicksand, which will defeat Mario if he lands in it, and the ghost enemy Boo.

The Boo is an enemy that (with normal parameters) chases Mario if he is looking away from it and is less than a certain distance away.
Once Mario gets too far, the Boo moves back to its original position. 
Unlike most enemies, jumping on a Boo does not kill it, but instead sends it a short distance forward or backward, which we will use to help Mario cross the quicksand. Some walls stop the Boo, but it can go through certain walls that normal Mario cannot go through; we call these Boo-only walls. The Boo is also unable to go through doors. We also make use of one-way walls which Mario and the Boo can go through in one direction but not the other.

For the setup, we use one Boo in Super Mario 64 DS and two Boos in Super Mario 64. Performing a kick while in the air sends Mario a short distance up and can normally only be performed once per jump.
But Mario can kick after jumping on a Boo
in Super Mario 64 DS even if he already kicked, allowing him to jump on the same Boo. This is not true in Super Mario 64, so jumping on a second Boo
is necessary to stall long enough to jump on the first Boo again.

\begin{theorem} Collecting a Star in a Super Mario 64/Super Mario 64 DS course is \PSPACE-hard assuming no course size limits.
\end{theorem}
\begin{proof}
We reduce from 1-player motion planning with the symmetric self-closing door (\thmref{thm:planar-scd-pspace}), where the target to reach is
a Star. \figref{fig:sm64ds-door} shows the simulation.

In the setup below, Mario goes from $1$ to $2$ and opens the $3 \to 4$ traversal
by going through the door on the bottom-left and hopping on the Boo(s) to the top-left.
Then Mario lets the Boo(s) chase him a little to turn the Boo(s), and hops on the Boo(s) to push them into the top-right.
Finally, Mario goes through the top-left door. Mario cannot just jump to the other side because the distance is too far.
He also cannot go into the traverse path because of the Boo-only wall. The Boo(s) will try to go back to their home position(s),
but cannot because they are stuck behind a one-way wall and a regular wall. If Mario does not move the Boo(s) to the top-right, they still
cannot get back to their home position(s) because of a different one-way wall, so Mario cannot leave the $1 \to 2$ traversal open.
%\looseness=-1

Mario goes from $3$ to $4$ by going through the top-right door and hopping on the Boo(s) to the bottom-right,
then going through the bottom-right door. The Boo(s) will go back to their original position(s) at the bottom left on their own.

Mario cannot lure the Boo(s) away from the gadget because it is completely walled in except for the doors, which the Boo(s) cannot go through.
\end{proof}

\begin{figure}
  \centering
  \includegraphics[scale=1]{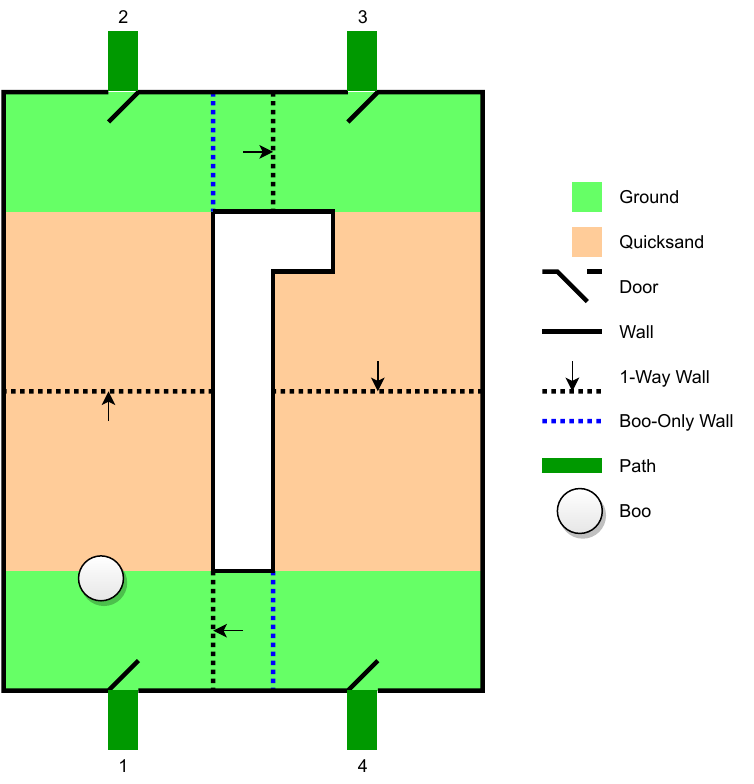}
  \caption{Simulation of a symmetric self-closing door in Super Mario 64 DS. In Super Mario 64, there are 2 Boos instead of 1.
  The ground and quicksand are on the same vertical level.
  The room is covered by a ceiling. The hallways are too wide to wall jump across.}
  \label{fig:sm64ds-door}
\end{figure}

\subsection{Super Mario Sunshine}
\label{sec:super-mario-sunshine}
Super Mario Sunshine is a 3D Mario game for the GameCube
where Mario is falsely accused of spreading graffiti and is forced to clean it up before he can leave.
Like Super Mario 64, this game includes one-way walls.
This game features a new device, F.L.U.D.D., attached to Mario's back that allows him to spray water.
Lily Pads float on water; the player can ride a Lily Pad and cause it to move by spraying water in the opposite direction.
Sludge is an environmental hazard which kills Mario if he touches it.
The general goal of a level is to collect Shrine Sprites.

\begin{theorem} Collecting a Shine Sprite in a Super Mario Sunshine level is \PSPACE-hard assuming no level size limits.
\end{theorem}
\begin{proof}
We reduce from 1-player motion planning with the symmetric self-closing door (\thmref{thm:planar-scd-pspace}), where the target to collect is a Shine Sprite.
\figref{fig:sms-door} shows the simulation of a symmetric self-closing door.

%\xxx{https://www.youtube.com/watch?v=jt6NpNbaLGc is Timestamp 33:29 "The stupid lilypad dissolves" an issue? This is why there is thin water above the sludge.}

The thin water above the sludge prevents the Lily Pad from disintegrating, while preventing Mario from crossing without using the Lily Pad.
Mario goes from $1$ to $2$ and opens
the $3 \to 4$ traversal by crossing
the one-way wall and riding the Lily Pad across, then moves the Lily Pad partially across the slit so it can be accessed from the other side.
He cannot leak to the section between $3$ and $4$ because the slits are too thin. The sludge is too long to simply jump to the other
side, so the Lily Pad is needed. Mario cannot do anything from $2$ because the one-way wall blocks him from going to $1$.
Mario goes from $3$ to $4$ in a similar manner.
\end{proof}

\begin{figure}
  \centering
  \includegraphics[scale=1]{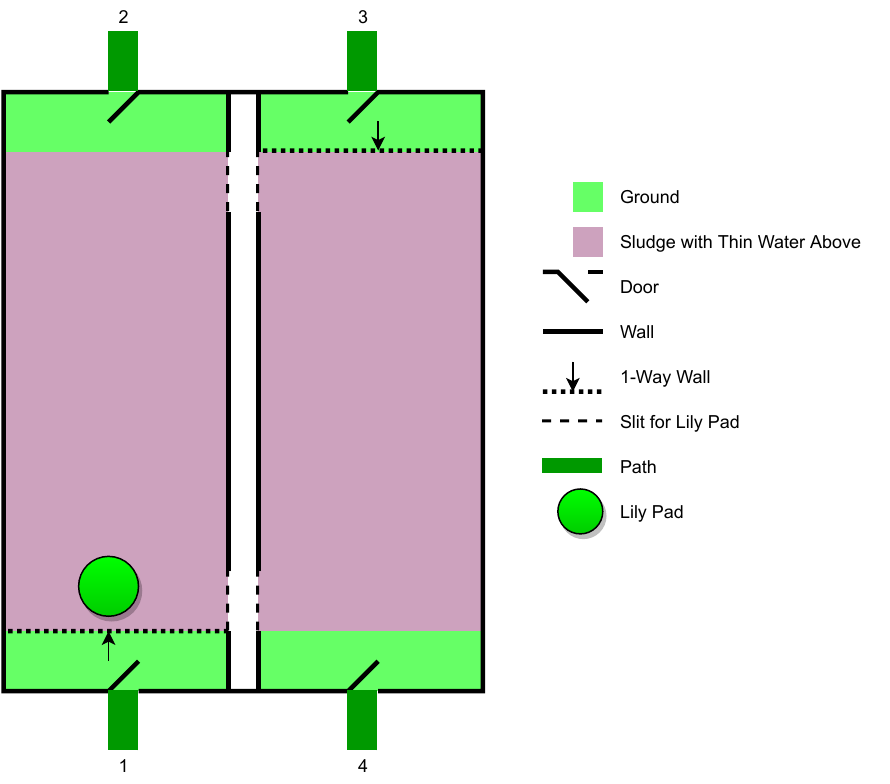}
  \caption{Simulation of a symmetric self-closing door in Super Mario Sunshine. The slits allow the Lily Pad to cross without allowing
  bulky Mario to do so. The hallways are too wide to wall jump across.}
  \label{fig:sms-door}
\end{figure}

\subsection{Super Mario Galaxy}
\label{sec:super-mario-galaxy}
Super Mario Galaxy is a 3D Mario game for the Wii
where Mario goes to space. He encounters alien creatures along the way and collects Power Stars to
restore the power of a spaceship. The game features downward gravity, upward gravity, sideways gravity, spherical gravity, cubical gravity,
tubular gravity, cylindrical gravity that allows infinite freefall, W-shaped gravity, gravity that cannot make up its mind, and most importantly,
controllable gravity.

Dark matter disintegrates Mario when he touches it, resulting in death.
The Gravity Switch changes the direction of gravity
when spun and can be spun multiple times. 

\begin{theorem} Collecting a Power Star in a Super Mario Galaxy galaxy is \PSPACE-hard assuming no galaxy size limits.
\end{theorem}
\begin{proof}
We reduce from 1-player motion planning with the symmetric self-closing door (\thmref{thm:planar-scd-pspace}), where the target to collect is a Power Star.
\figref{fig:smg-door} shows the simulation of a symmetric self-closing door.

The Gravity Switch in this construction switches gravity between down and up. Mario goes from $1$ to $2$ by crossing the one-way wall and hitting the Gravity Switch on his way to the right.
This is forced because of a pit of dark matter, and closes
the $1 \to 2$ traversal because when gravity points up, attempting the traversal would land Mario on dark matter. At the same time,
it opens the $3 \to 4$ traversal. Mario cannot enter $2$ and do anything useful because flipping the Gravity Switch means
falling in the pit of dark matter. Mario goes from $3$ to $4$ in a similar manner.
\end{proof}

\begin{figure}
  \centering
  \includegraphics[width=\linewidth]{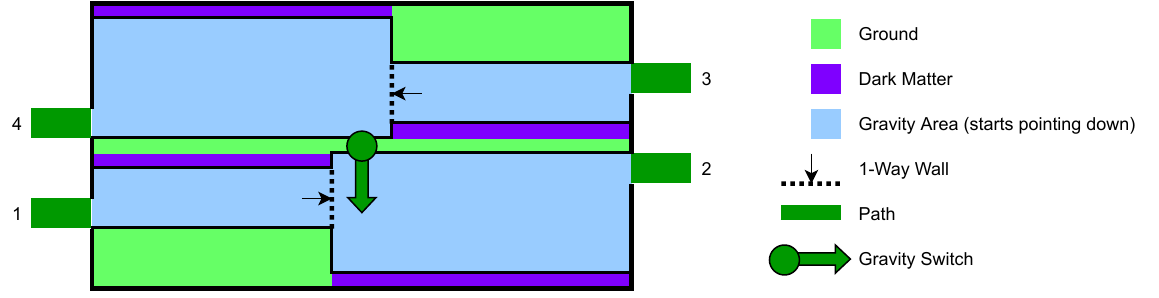}
  \caption{Simulation of a symmetric self-closing door in Super Mario Galaxy. This is a side view and is essentially
  2-dimensional.}
  \label{fig:smg-door}
\end{figure}

\subsection{Super Mario Galaxy 2}
\label{sec:super-mario-galaxy-2}
Super Mario Galaxy 2 is the sequel to Super Mario Galaxy (also for the Wii)
which features new galaxies. Although similar in gameplay, each game has some objects which do not appear in the other. This reduction is very similar to that in Section~\ref{sec:super-mario-sunshine} with the Lily Pad and sludge.

The Leaf Raft is a raft that floats on water and that can be moved by standing on its edge. Lava is an environmental hazard which damages Mario. Unrealistically, we can have a thin layer of water on top of a layer of lava. Finally, an electric fence is another environmental hazard which damages Mario but will allow the Leaf Raft to pass through it.

\begin{theorem} Collecting a Power Star in a Super Mario Galaxy 2 galaxy is \PSPACE-hard assuming no galaxy size limits.
\end{theorem}
\begin{proof}
We reduce from 1-player motion planning with the symmetric self-closing door (\thmref{thm:planar-scd-pspace}), where the target to collect is a Power Star.
\figref{fig:smg2-door} shows the simulation of a symmetric self-closing door. The Power Star is under a Daredevil Comet, making Mario have
only 1 HP, so he cannot afford to bounce in the lava or shock-boost through an electric fence.

 Mario goes from $1$ to $2$ and opens
the $3 \to 4$ traversal by crossing
the one-way wall and riding the Leaf Raft across, then carefully making the Leaf Raft partially cross the electric fence.
He cannot move to the section between $3$ and $4$ because of the electric fences. The lava is too long to simply jump to the other
side, so the Leaf Raft is needed. Mario cannot do anything from $2$ because the one-way wall blocks him from going to $1$.
Mario goes from $3$ to $4$ in a similar manner.
\end{proof}

\begin{figure}
  \centering
  \includegraphics[scale=1]{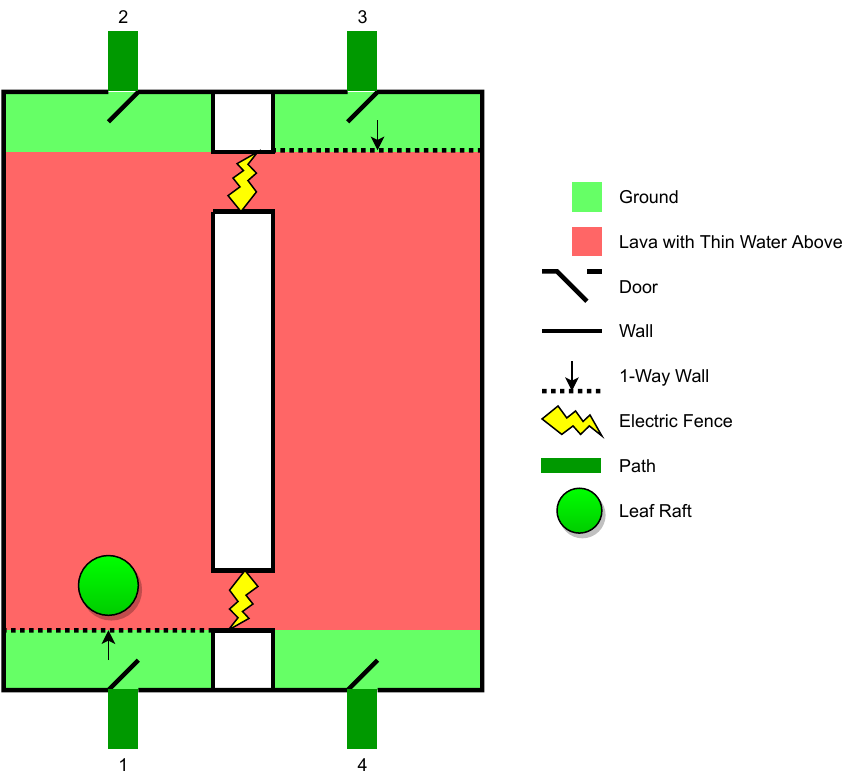}
  \caption{Simulation of a symmetric self-closing door in Super Mario Galaxy 2. The thin water above the lava allows the Leaf Raft
  to float but does not allow Mario to swim in it without taking damage from the lava beneath. The walls cannot be wall jumped on (Super Mario Galaxy 2
  allows vertical walls that cannot be wall jumped on).}
  \label{fig:smg2-door}
\end{figure}

\subsection{Super Mario 3D Land and Super Mario 3D World}
\label{sec:super-mario-3d-land-world}
Super Mario 3D Land and Super Mario 3D World are 3D Mario games
for the 3DS and Wii U, respectively,
that are based on the New Super Mario Bros.\ series instead of earlier
3D Mario games. Instead of collecting Shine Sprites or Stars, the player traverses a level to reach the flagpole at the end. In addition,
the player does not have a health bar but loses their powerup or dies when taking damage. Levels have time limits.

The Switchboard is a platform that rides on tracks and contains two arrows. If Mario steps on an arrow, the Switchboard goes in the direction of said arrow. Otherwise, the Switchboard does not move. In Super Mario 3D World, the Switchboard can be controlled by using the Wii U gamepad, but only if the Switchboard is visible. We also make use of a pit deep enough that Mario cannot jump out.

\begin{theorem} Reaching the flagpole at the end of a Super Mario 3D Land/World level is \PSPACE-hard assuming no level size limits
and no time limit.
\end{theorem}
\begin{proof}
We reduce from 1-player motion planning with the symmetric self-closing door (\thmref{thm:planar-scd-pspace}), where the target to reach is
the flagpole. Figures~\ref{fig:sm3dl-door} and~\ref{fig:sm3dl-door-side} show the simulation of a symmetric self-closing door.

In the setup below, Mario goes from $1$ to $2$ and opens the $3 \to 4$ traversal by going through the tunnel,
then riding the Switchboard to the other side
making sure it goes through the wall, then going
through the tunnel on the other side. Mario cannot move the Switchboard to the other side and then leave via $1$ because the pit
is too wide and the Switchboard cannot be moved without going through the (one-way) tunnel because it is blocked by a wall. If the Switchboard
is on the wrong side, it cannot be moved either backward (because the path stops) or forward (because a wall then blocks the way). This
ensures that Mario can go from $1$ to $2$ if and only if the Switchboard is already at $1$, and then the Switchboard must stay at $2$/$3$.
Mario goes from $3$ to $4$ in a similar manner.
\end{proof}

\begin{figure}
  \centering
  \includegraphics[scale=1]{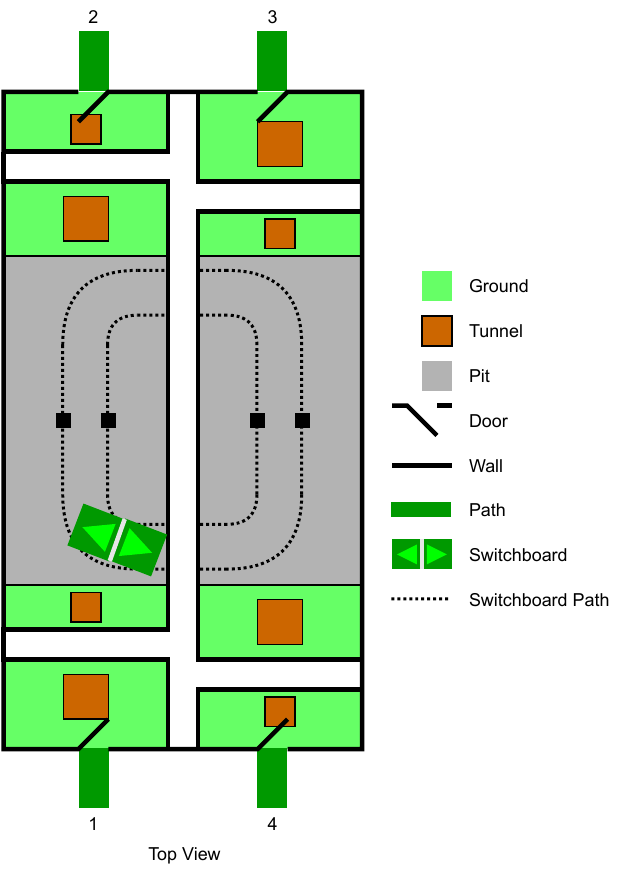}
  \caption{Top view of the simulation of a symmetric self-closing door in Super Mario 3D Land/World. The pit is long enough for the
  player to not be able to jump from the ground to anywhere near the center of the pit. The wider sides of the tunnels are wide enough to
  not allow wall jumping, making the tunnels one-way. The hallways are also too wide to wall jump across.}
  \label{fig:sm3dl-door}
\end{figure}

\begin{figure}
  \centering
  \includegraphics[scale=1]{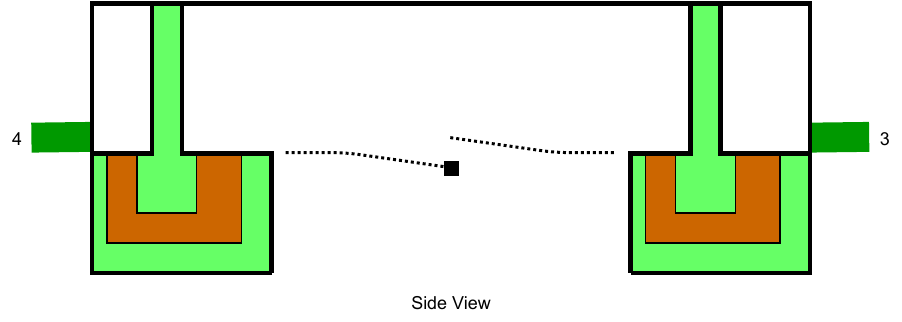}
  \caption{Side view of the simulation of a symmetric self-closing door in Super Mario 3D Land/World.}
  \label{fig:sm3dl-door-side}
\end{figure}

\subsection{Super Mario Odyssey}
\label{sec:super-mario-odyssey}
Super Mario Odyssey is a 3D Mario game for the Switch
where Mario travels to different kingdoms collecting Power Moons and eventually goes to the Moon.
Mario has the ability (via his hat Cappy) to capture certain enemies and objects to use their powers, but such objects tend to reset their position after being uncaptured,
so we will not be using them here.

We make use of a Jaxi, poison, and timed platforms. A Jaxi is a statue lion that can be ridden safely across poison, which is a hazard that kills Mario.
%(Riding a Jaxi into a pit causes the Jaxi to respawn, so we wall off the Jaxi area so that there are no pits.)
%\xxx{What exactly is a pit and how does it appear here? I don't see it in the diagram, looks like the Jaxi just goes back and forth...}
%\xxx[Erik]{There are no pits in the construction, so no need to mention it.}
A timed switch makes some event happen for a specific amount of time.
In our reduction, timed switch X makes platform X appear for just long enough
for Mario to make a traversal.

\begin{theorem} Collecting a Power Moon in a Super Mario Odyssey kingdom is \PSPACE-hard assuming no kingdom size limit.
\end{theorem}
\begin{proof}
We reduce from 1-player motion planning with the symmetric self-closing door (\thmref{thm:planar-scd-pspace}), where the target to reach is
a Power Moon.
\figref{fig:smo-door} shows the simulation of a symmetric self-closing door.

Mario goes from $1$ to $2$ by pressing timed switch A, riding the Jaxi to the right, and traversing platform A. This
opens the $3 \to 4$ traversal while closing the $1 \to 2$ traversal. Mario cannot
go to $3$ because of the wide gap, or to $4$ because platform B is gone. The Jaxi is required because the poison it is on is
very wide. Mario cannot do anything useful if he tries to enter from $2$ or $4$ because the platforms would be gone. Mario
goes from $3$ to $4$ in a similar manner.
\end{proof}

\begin{figure}
  \centering
  \includegraphics[width=\linewidth]{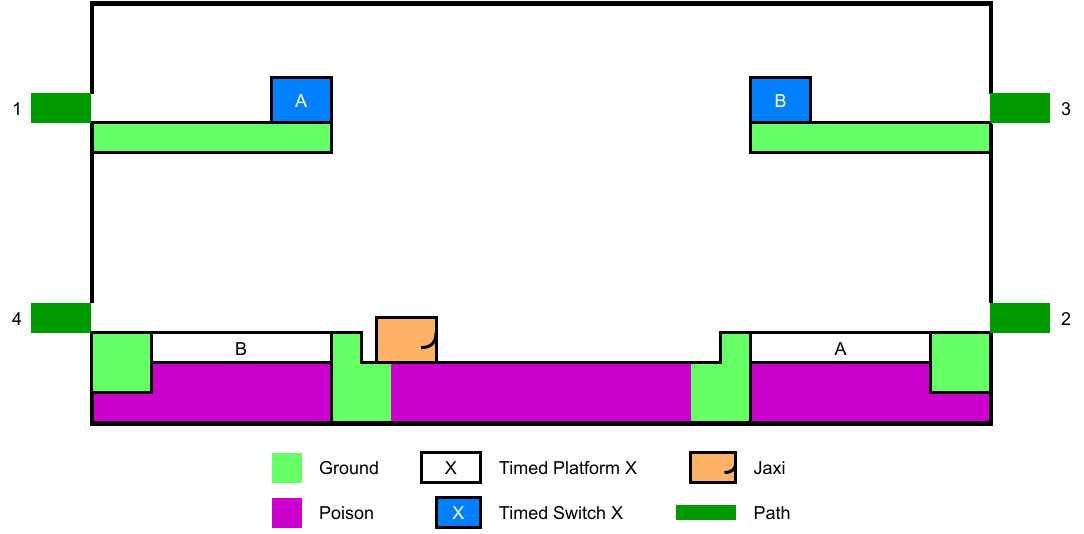}
  \caption{Simulation of a symmetric self-closing door in Super Mario Odyssey. This is a side view and is essentially
  2-dimensional. All strips of poison are way too wide for Mario to cross with his various aerial skills, and the platforms with timed
  switches are too high to get to from below.}
  \label{fig:smo-door}
\end{figure}

\subsection{Captain Toad:\ Treasure Tracker}
\label{sec:captain-toad}
Captain Toad:\ Treasure Tracker is a 3D puzzle platformer in the Mario universe,
originally appearing as a type of level in Super Mario 3D World, and then
released as a stand-alone game on the Wii U and ported to the 3DS and Switch.
Notably, Toad can fall but not jump.
The game contains rotating platforms controlled by a wheel that Toad must be adjacent to in order to move. The platforms move in $90^\circ$ increments. We show \PSPACE-hardness by constructing an antiparallel symmetric self-closing door (Theorem~\ref{thm:planar-scd-pspace}).

\begin{theorem} Collecting Stars in Captain Toad:\ Treasure Tracker is \PSPACE-hard assuming no level size limit.
\end{theorem}
\begin{proof}
\figref{fig:ToadSSCD} gives a top-down view of the construction. There is a U-shaped rotating platform at a height slightly below the high ground and far above the low ground. The U-shaped platform rotates counterclockwise and can be reached from the nearby high ground; however, the gap between the back of the U and the other side is too wide for Toad to cross. Further, the dividing wall sits slightly above the rotating platform, preventing Toad from crossing. Toad is able to go onto the U platform from the high ground, activate the gear twice, and step off the U platform onto the low ground across the gap. The U platform is now facing the other way, allowing Toad to enter from the high ground on the other side, but preventing other traversals.
\end{proof}

\begin{figure}
  \centering
  \includegraphics[width=.5\textwidth]{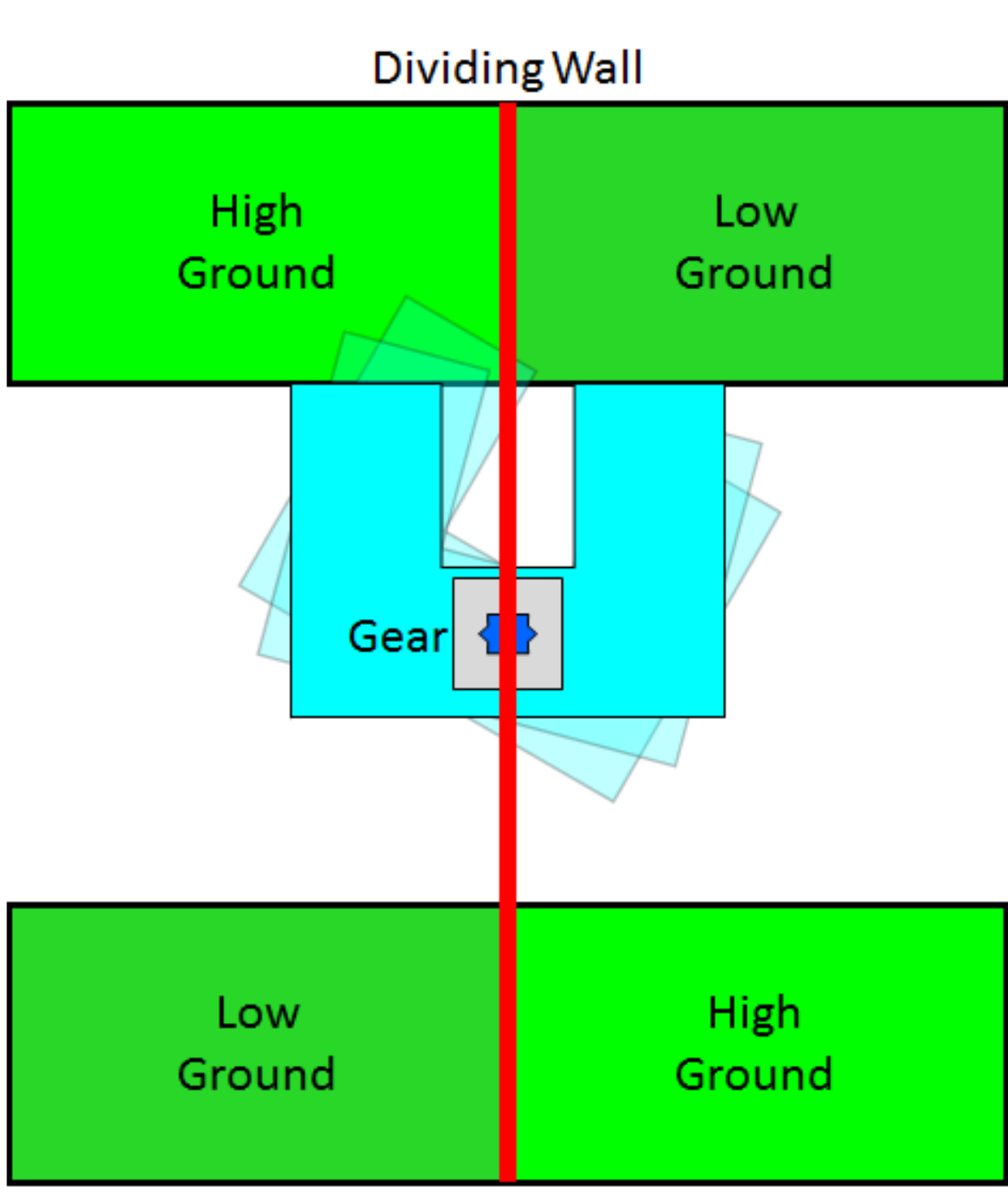}
  \caption{Top view of a simulation of a symmetric self-closing door.}
  \label{fig:ToadSSCD}
\end{figure}
%pswithes and binary sliding platforms sound obviously hard also.
%\xxx{Section on Nintendo \PSPACE proofs no longer needing crossovers}

\section{Open Problems}

We mention a few directions for future work.

\paragraph{Symmetric open--close door.}
Define a \defn{symmetric open--close door} to be a gadget with four traversals:
\emph{opening} and \emph{closing} buttons
set the state to open and closed, respectively;
the \emph{traverse} tunnel can be traversed only in the open state; and
the \emph{antitraverse} tunnel can be traversed only in the closed state.
Like the symmetric self-closing door, the symmetric open--close door is symmetric
under an ``open''/``close'' swap.
(We focus on opening and closing buttons because if either were a tunnel,
we could ignore one of the traverse tunnels and obtain an open--close door.)
Clearly the symmetric open--close door is universal:
if we join the closing button to the entrance of the antitraverse
tunnel, then we obtain (the closing tunnel of) an open--close door.
But for which cyclic orderings of the six locations is it planarly universal?

While the symmetric open--close door may seem more complicated,
it is in some cases simpler to build.
Take the Super Mario Bros.\ door \cite[Figure~6]{demaine2016super},
reproduced in \figref{fig:smb-door}.
Arguably, the right half of the gadget serves two distinct functions,
and it is simpler to separate them like in the left half.
The result is the symmetric open--close door of \figref{fig:smb-door-symmetric},
which also happens to be symmetric under a vertical reflection.
It seems more natural to find the symmetric open--close door in this case,
and then derive an open--close door from it.
Thus a more thorough analysis of symmetric open--close doors seems warranted.

\begin{figure}
  \centering
  \footnotesize

  \subcaptionbox{\label{fig:smb-door} Original door from \cite[Figure~6]{demaine2016super}.}{%
    \hspace*{4.1em}%
    \begin{overpic}[scale=0.75]{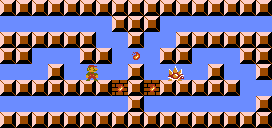}
      \put(11,38){\makebox(0,0)[r]{\strut traverse out}}
      \put(11,20.25){\makebox(0,0)[r]{\strut traverse in}}
      \put(11,8.5){\makebox(0,0)[r]{\strut opening button}}
      \put(90,38){\makebox(0,0)[l]{\strut closing out}}
      \put(90,14.25){\makebox(0,0)[l]{\strut closing in}}
    \end{overpic}%
    \hspace*{2.6em}%
  }%
  \hfill
  \subcaptionbox{\label{fig:smb-door-symmetric} Symmetric open--close door.}{%
    \hspace*{4.1em}%
    \begin{overpic}[scale=0.75]{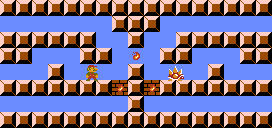}
      \put(11,38){\makebox(0,0)[r]{\strut traverse out}}
      \put(11,20.25){\makebox(0,0)[r]{\strut traverse in}}
      \put(11,8.5){\makebox(0,0)[r]{\strut opening button}}
      \put(90,38){\makebox(0,0)[l]{\strut antitraverse out}}
      \put(90,20.25){\makebox(0,0)[l]{\strut antitraverse in}}
      \put(90,8.5){\makebox(0,0)[l]{\strut closing button}}
    \end{overpic}%
    \hspace*{4.8em}%
  }
  \caption{The Super Mario Bros.\ open--close door (a) can be derived from the symmetric open--close door (b) by joining the closing button to antitraverse in.}
  \label{fig:smb-doors}
\end{figure}

\paragraph{Reflections.}
Our paper,
and all other papers developing the motion-planning-through-gadgets framework,
treat each gadget as the same as its reflection.
This assumption is motivated by applications to games, which typically
do not significantly distinguish between ``left'' and ``right''
(whereas ``up'' and ``down'' are often distinguished, because of gravity),
so if you can build one gadget, you can mirror it through a vertical line
to obtain its reflection.
Nonetheless, it would be interesting to determine whether all door gadgets
are universal even when forbidding reflection.

\paragraph{2-player.}
Zhang \cite{janggi} developed a \emph{2-player} door framework
for proving EXPTIME-hardness of 2-player unbounded games,
specifically Chinese and Korean chess.
They considered the open--close door (basing it on \cite{demaine2016super},
as it predated the present work \cite{doors-fun2020}).
In their framework, each traversal of the door can be colored
according to which player (red or blue) can traverse it.
We could also consider doors where some traversals can be traversed by
either player (purple).
In addition, we can now consider open--close doors where the opening traversal
\emph{optionally} opens the door (up to the player's choice),
and similarly for the closing traversal, as we can no longer assume
that the player always wants to open a door if it can.
Closing-buttoned open--close doors also now make sense.
This opens a whole variety of door gadgets.
Which of them suffice for EXPTIME-hardness?

\section*{Acknowledgments}

This work was initiated during open problem solving in the MIT class on
Algorithmic Lower Bounds: Fun with Hardness Proofs (6.892)
taught by Erik Demaine in Spring 2019.
We thank the other participants of that class
for related discussions and providing an inspiring atmosphere.

%\nocite{*} % list references even if not cited
\bibliographystyle{alpha}
\bibliography{thebib}

\end{document}